\documentclass[12pt,twoside]{amsart}
\usepackage{amsmath,amsthm,amsmath,amsrefs}

\usepackage{calc}
\usepackage{setspace}
\usepackage{amsbsy,amsmath,amsfonts,amsthm,amssymb,color}
\usepackage[margin=1in]{geometry}
\numberwithin{equation}{section}
\theoremstyle{plain}
\newtheorem{theorem}{Theorem}[section] 
\newtheorem{lemma}[theorem]{Lemma}
\newtheorem{example}[theorem]{Example}
\newtheorem{definition}[theorem]{Definition}
\newtheorem{corollary}[theorem]{Corollary}
\newtheorem{remark}[theorem]{Remark}
\newtheorem{proposition}[theorem]{Proposition}

\theoremstyle{definition}

\usepackage{mathrsfs}

\newcommand{\bN}{\mathbb{N}}

\newcommand{\cF}{\mathcal{F}}

\newcommand{\cC}{\mathcal{C}}

\newcommand{\cB}{\mathcal{B}}

\newcommand{\la}{\langle}
\newcommand{\ra}{\rangle}

\newcommand{\cD}{\mathcal{D}}

\newcommand{\cH}{\mathcal{H}}
\newcommand{\cA}{\mathcal{A}}

\newcommand{\cU}{\mathcal{U}}
\newcommand{\bC}{\mathbb{C}}
\newcommand{\bR}{\mathbb{R}}
\newcommand{\cG}{\mathcal{G}}
\newcommand{\cR}{\mathcal{R}}
\newcommand{\id}{\text{id}}

\newcommand{\cI}{\mathcal{I}}

\newcommand{\cE}{\mathcal{E}}

\newcommand{\tr}{\text{tr}}

\usepackage{tikz}
\usepackage{tikz-cd}

\title[Universality of graph homomorphism games]{Universality of graph homomorphism games and the quantum coloring problem}

\author{Samuel J. Harris}

\address{Northern Arizona University\\
Department of Mathematics \& Statistics\\
801 S. Osborne Dr.\\
Flagstaff, AZ\\
86011 USA}
\email{samuel.harris@nau.edu}

\begin{document}

\begin{abstract}
We show that quantum graph parameters for finite, simple, undirected graphs encode winning strategies for all possible synchronous non-local games. Given a synchronous game $\cG=(I,O,\lambda)$ with $|I|=n$ and $|O|=k$, we demonstrate what we call a weak $*$-equivalence between $\cG$ and a $3$-coloring game on a graph with at most $3+n+9n(k-2)+6|\lambda^{-1}(\{0\})|$ vertices, strengthening and simplifying work implied by Z. Ji \cite{Ji13} for winning quantum strategies for synchronous non-local games. As an application, we obtain a quantum version of L. Lov\'{a}sz's reduction \cite{Lo73} of the $k$-coloring problem for a graph $G$ with $n$ vertices and $m$ edges to the $3$-coloring problem for a graph with $3+n+9n(k-2)+6mk$ vertices. Moreover, winning strategies for a synchronous game $\cG$ can be transformed into winning strategies for an associated graph coloring game, where the strategies exhibit perfect zero knowledge for an honest verifier. We also show that, for ``graph of the game" $X(\cG)$ associated to $\cG$ from A. Atserias et al \cite{AMRSSV19}, the independence number game $\text{Hom}(K_{|I|},\overline{X(\cG)})$ is hereditarily $*$-equivalent to $\cG$, so that the possibility of winning strategies is the same in both games for all models, except the game algebra. Thus, the quantum versions of the chromatic number, independence number and clique number encode winning strategies for all synchronous games in all quantum models.
\end{abstract}

\maketitle

\section{Introduction}

Non-local games, and in particular, synchronous non-local games, have been the source of a great deal of study in recent years, due to their direct relation to the nature of entanglement, quantum computing, and quantum cryptography. Such a game involves two players (Alice and Bob) that are working cooperatively to win a multi-round game as many times as possible, by providing correct answers to questions posed by a neutral referee. The players are not allowed to communicate once the game begins, but can agree upon a strategy beforehand. Such games can be constructed where the players are unable to win the game with any classical strategy, but can win the game using quantum entanglement. Formally, a game is given by $\cG=(I,O,\lambda)$, where $I$ is a finite input set, $O$ is a finite output set, and $\lambda:O \times O \times I \times I \to \{0,1\}$ is a rule function. Each round of the game, each player receives a question from the set $I$, and must respond with an answer from the set $O$. If their questions and answers satisfy $\lambda(a,b,x,y)=1$, then they win the round of the game; otherwise they lose.

Examples of such games where the players can only win using entanglement are of great interest in the field and have been used to strengthen our understanding of the different models of bipartite correlations. The strategy that Alice and Bob uses for the game involves the probability $p(a,b|x,y)$ that Alice and Bob answer $a$ and $b$, respectively, given that they received the questions $x$ and $y$, respectively. In each model, the possible probability distributions in an $n$-input, $k$-output system are modelled by the elements of $\bR^{n^2k^2}$ of the form 
\[(p(a,b|x,y))_{1 \leq x,y \leq n, \, 1 \leq a,b \leq k}=\la E_{a,x}F_{b,y}\psi,\psi \ra,\] where $\psi$ is a pure state (unit vector) in a Hilbert space $\cH$; for each $1 \leq x \leq n$ and $1 \leq y \leq n$, the operators $\{E_{a,x}\}_{a=1}^k$ (Alice's measurement system for input $x$) and $\{F_{b,y}\}_{b=1}^k$ (Bob's measurement system for input $y$) are positive in $\mathcal{B}(\mathcal{H})$ and $\displaystyle \sum_{a=1}^k E_{a,x}=\sum_{b=1}^k F_{b,y}=I_{\cH}$, and $[E_{a,x},F_{b,y}]=0$ for all $a,b,x,y$. Our convention is that the inner product above is linear on the left, and conjugate-linear on the right. The four models that are usually considered in an $n$-input, $k$-output system, are given by restrictions on what probability densities above are allowed. The ``loc" model (for local, or classical) is given by imposing the extra condition that all the measurement operators of both players commute; the set of such correlations is denoted by $C_{loc}(n,k)$. The ``q" model (for quantum, assuming a finite-dimensional entanglement resource space) instead assumes that $\cH$ is finite-dimensional and decomposes as $\cH=\cH_A \otimes \cH_B$, and that each $E_{a,x}$ acts as the identity on $\cH_B$, and each $F_{b,y}$ acts as the identity on $\cH_A$; the set of such correlations is denoted by $C_q(n,k)$. The ``qa" model (for quantum approximate) is the collection of all probability densities that are pointwise limits of those in $C_q(n,k)$; in other words, $C_{qa}(n,k)=\overline{C_q(n,k)}$. The set of all possible correlations is given by the ``qc" model (for quantum commuting, arising from quantum field theory); the set of such correlations is denoted by $C_{qc}(n,k)$. For each of the sets $C_t(n,k)$ and a non-local game $\cG=(I,O,\lambda)$, a \textbf{winning $t$-strategy} for $\cG$ is given by an element $(p(a,b|x,y)) \in C_t(n,k)$ satisfying $p(a,b|x,y)=0$ whenever $\lambda(a,b,x,y)=0$. If such a strategy exists, then the players win the game using that strategy with probability $1$. Finding games where a winning strategy exists in one model, but does not exist in a more restrictive model, is a common method of exhibiting new separations between the models. It is well-known that $C_{loc}(n,k) \subseteq C_q(n,k) \subseteq C_{qa}(n,k) \subseteq C_{qc}(n,k)$, and all three of these containments are known to be distinct by using non-local games \cite{CHSH69,Slo19,JNVWY20}.

For \textbf{synchronous games}--that is, those satisfying $\lambda(a,b,x,x)=0$ for $a \neq b$, the structure of winning strategies is refined by a theorem of V. Paulsen et al \cite[Theorem~5.5]{PSSTW16}: any winning strategy $(p(a,b|x,y))$ in $C_{qc}(n,k)$ for a synchronous game $\cG=(I,O,\lambda)$ with $|I|=n$ and $|O|=k$ can be given by a tracial state $\tau$ on a unital $C^*$-algebra $\mathcal{A}$ (that is, $\tau$ is a positive linear functional with $\tau(1_{\cA})=1$ and $\tau(XY)=\tau(YX)$ for all $X,Y \in \mathcal{A}$), and projections $E_{a,x} \in \cA$ satisfying $\sum_{a=1}^k E_{a,x}=1_{\cA}$ for all $1 \leq x \leq n$ and $p(a,b|x,y)=\tau(E_{a,x}E_{b,y})$. Moreover, $\tau$ can always be arranged to be faithful. (Such correlations can be viewed as the players having a ``quantum shared function"; indeed, in the case when $p \in C_{loc}(n,k)$ and $p$ is deterministic--that is, $p(a,b|x,y) \in \{0,1\}$ for all $a,b,x,y$--then $p$ does arise from a shared function.) For $loc$ synchronous strategies, one can arrange to have $\cA$ abelian, while for $q$ synchronous strategies, one can arrange to have $\cA$ finite-dimensional \cite{HMPS19}. (In fact, for winning strategies for synchronous games, an extreme point argument allows for $\cA$ to simply be a matrix algebra.) For $t=qa$, one can arrange for $\cA=\cR^{\cU}$, a tracial ultrapower of the hyperfinite $II_1$-factor von Neumann algebra $\cR$ with respect to a free ultrafilter $\cU$ on $\bN$, using the unique trace on $\cR$ \cite{KPS18}. These tracial descriptions have important applications in the theory of operator algebras. One additional advantage of synchronous games is that, because only one player's operators need to be considered, it is possible to define an associated game $*$-algebra $\cA(\cG)$: a universal, unital $*$-algebra whose representations encode the existence of winning strategies for $\cG$ in the various models \cite{HMPS19}.

Despite recent advances in our understanding of how entanglement can be used in non-local games and how the models are distinguished from each other, there is still much work being done on understanding how these separations arise and how entanglement truly varies in the different models. For example, while it is known that the $qa$ and $qc$ models are not the same, due to the significant result that $\text{MIP}^*=\text{RE}$ \cite{JNVWY20}, the ramifications of this result are widespread and still being worked out. Even the known separations of $q$ and $qa$ that arise from synchronous games are very large \cite{KPS18}.

To this end, it is helpful to know what classes of synchronous games are ``universal", in the sense that every synchronous game can be transformed into a game from this class, while preserving representations of the game $*$-algebra, in some sense. The idea behind such equivalences is to preserve the existence of winning strategies in the different models. The most common tool for such a transformation is $*$-equivalence. Two games $\cG_1$ and $\cG_2$ are $*$-equivalent if there are unital $*$-homomorphisms $\cA(\cG_1) \to \cA(\cG_2)$ and $\cA(\cG_2) \to \cA(\cG_1)$ \cite{KPS18,HMPS19}. If synchronous games $\cG_1$ and $\cG_2$ are $*$-equivalent, then for $t \in \{loc,q,qa,qc\}$, $\cG_1$ has a winning $t$-strategy if and only if $\cG_2$ does. For example, synchronous non-local games with $3$ answers are universal \cite{Fr20,H22}: every synchronous game is $*$-equivalent to a synchronous game with $3$ answers. The class of bisynchronous games is universal \cite{PR21}; later, it was shown that this even holds with equal question and answer sets in \cite{H22}. The class of graph homomorphism games $\text{Hom}(G,H)$ between two finite simple undirected graphs, exhibiting some of the known peculiarities of synchronous games (see \cite{MR16,HMPS19}) has been thought to be universal. Such equivalences between games can help simplify the over-arching structure of synchronous games that separate the distinct models of quantum mechanics. Equivalences between classes of games are also helpful for understanding their properties, even if such games are not universal. For example, the graph isomorphism game $\text{Iso}(G,H)$, where $G$ and $H$ have the same number of vertices, cannot be universal, since the existence of the game algebra $\cA(\text{Iso}(G,H))$ automatically forces the existence of a winning strategy in $C_{qc}(n,k)$ \cite{BCEHPSW20}. In contrast, the homomorphism game $\text{Hom}(K_5,K_4)$ does not have this property \cite{HMPS19}, and cannot be equivalent to a graph isomorphism game. Similarly, the synchronous linear binary constraint system game $\text{syncBCS}(A,b)$ cannot be universal \cite{BCEHPSW20}. Nevertheless, every synchronous linear binary constraint system game is $*$-equivalent to a graph isomorphism game \cite{BCEHPSW20,G21}.

In this paper, we prove that three of the so-called graph parameter games are universal. The main reuslt is that every synchronous non-local game is equivalent, in a weak sense, to a $3$-coloring game of an associated graph. We also construct hereditary $*$-equivalences of synchronous games to games involving the independence number and the clique number of other graphs associated with the original game that arise in \cite{AMRSSV19}. As a result of our work, we obtain a quantum version of L. Lov\'{a}sz's reduction theorem of the $k$-coloring problem to the $3$-coloring problem \cite{Lo73}. Another immediate application is that the $qa$ and $qc$ versions of the chromatic number of a graph are distinct, which was not previously known. The analogous statements for independence number and clique number also follow from this paper. Moreover, each of the quantum chromatic number $\chi_t$, quantum independence number $\alpha_t$, and quantum clique number $\omega_t$ are sufficiently rich enough quantities to detect the existence of winning strategies for any synchronous non-local game. In other words, for a synchronous game $\cG=(I,O,\lambda)$ and for any of the four models loc, q, qa and qc, there are two associated graphs $G$ and $H$ with the following properties:
\begin{itemize}
\item $\cG$ has a winning strategy in $C_t(n,k)$ if and only if $\chi_t(G)=3$; and
\item $\cG$ has a winning strategy in $C_t(n,k)$ if and only if $\alpha_t(H)=|I|=\omega_t(\overline{H})$.
\end{itemize}

As an application, any synchronous game $\cG$ with a winning $t$-strategy (where $t \in \{loc,q,qa,qc\}$) can be transformed into a game with a winning $t$-strategy that exhibits perfect zero knowledge for an honest verifier (see Proposition \ref{proposition: perfect zero knowledge} and the discussion thereafter). Such a strategy is one where, if the referee only asks questions where the players might lose, then probabilities for the answers given by the players from a winning strategy do not yield any information to the referee about the strategy used. The only thing that the referee can deduce is that the players possess a winning $t$-strategy for the game. This is helpful if the players want to demonstrate that they can perform a certain protocol (for example, in cryptography), without revealing to the referee what kind of entanglement and measurement operators they used. The concept of verifying statements while exhibiting zero knowledge (often referred to as ``zero knowledge proofs")  originated in \cite{GMR89}; we refer the reader to \cite{GSY19} for more information on perfect zero knowledge in the context of interactive provers.

Our method of transforming a synchronous game $\cG$ into a $3$-coloring game for an associated graph $G_{\lambda}$ involves a graph whose 3-coloring game roughly keeps track of sums of projections. One of our main tools to this end is the $3 \times 3$ rook's graph (see Figure \ref{figure: 3x3 rook's graph}), with two extra vertices that form a ``control" triangle with the $(1,2)$ vertex of the $3 \times 3$ rook's graph. We use representations of $A(Hom(G_{\lambda},K_3))$ and consider the corner given by a product of projections corresponding to colors that are used for the control triangle. This approach of cutting down by a projection first appeared in work of Z. Ji \cite{Ji13}, although the subgraph that we use as the ``gadget" here is simplified, and as a result, the graph obtained is much smaller. We rely heavily on the structure of $3 \times 3$ quantum permutations. A \textbf{quantum permutation} over a unital $*$-algebra $\cA$ is an $n \times n$ matrix $P=(p_{ij})$ with entries in $\cA$ such that each $p_{ij}$ is a self-adjoint idempotent, $p_{ij}p_{ik}=0$ for $j \neq k$ and $p_{ij}p_{kj}=0$ for $i \neq k$, and $\sum_{i=1}^n p_{ij}=\sum_{j=1}^n p_{ij}=1$. The entries of $3 \times 3$ quantum permutations automatically commute, even in the $*$-algebra setting (see \ref{proposition: 3x3 quantum permutation commutes}). While such a coloring game construction could be done in the hereditary game $*$-algebra setting for transforming synchronous games into $k$-coloring games for $k>3$, this becomes more complicated as entries of a $k \times k$ quantum permutation need not commute as soon as $k \geq 4$. Moreover, in $3$-coloring games, projections corresponding to adjacent vertices commute (see Proposition \ref{proposition: key three projection proposition, orthogonal version}), and this feature also arises from $3\times 3$ quantum permutations having pairwise commuting entries. In light of these technicalities, we restrict ourselves to the $3$-coloring setting here.

Our method for transforming $\cG$ into a game involving the independence number involves a construction from A. Atserias et al \cite{AMRSSV19} that uses a canonical graph associated to the game $\cG$, with vertex set $O \times I$ and edges given by the disallowed $4$-tuples from the rule function $\lambda$. While the equivalence for $3$-coloring games involves the game $*$-algebra, transforming $\cG$ into an independence number game requires a positivity argument, and can only be done at the level of the hereditary game $*$-algebra. This restriction is necessary as soon as $|I| \geq 4$ (see Remark \ref{remark: alg independence infinity} for more details).

It would be interesting to determine whether the game equivalences in this paper can be extended to approximate strategies. For example, suppose that $\cG$ is a synchronous non-local game, and that $G_{\lambda}$ is the associated graph in Theorem \ref{theorem: coloring game weakly subequivalent to original game} such that $\cG$ is weakly $*$-equivalent to the $3$-coloring game for $G_{\lambda}$. If the players can win the game $\cG$ with probability close to $1$, it is not clear whether their strategy can be transformed into a strategy that wins the $3$-coloring game for $G_{\lambda}$ with probability close to $1$ (and vice versa). We leave this problem open for future work.

The remainder of the paper is organized as follows. In Section \ref{section: preliminaries}, we recall the basics of the game $*$-algebra associated with a synchronous game, and introduce weak $*$-subequivalence and weak $*$-equivalence of games. In Section \ref{section: 3x3 rook's graph}, we prove that the game algebra of the $3$-coloring game of the $3 \times 3$ rook's graph is abelian (Corollary \ref{corollary: 3 coloring of 3x3 rook is abelian}). In Section \ref{section: equivalence}, we show that, to each synchronous game $\cG=(I,O,\lambda)$ with at least two inputs and $3$ outputs, there is an associated graph $G_{\lambda}$ so that the $3$-coloring game for $G_{\lambda}$ is weakly $*$-subequivalent to $\cG$ (Theorem \ref{theorem: coloring game weakly subequivalent to original game}). Conversely, $\cG$ is $*$-subequivalent to the $3$-coloring game for $G_{\lambda}$, making the two games weakly $*$-equivalent (Theorem \ref{theorem: game subequivalent to coloring}). In particular, for any $t$, $\cG$ has a winning $t$-strategy if if the $t$-coloring number of $G_{\lambda}$ is equal to $3$. We also discuss the application to perfect zero knowledge in this section. Lastly, in Section \ref{section: independence and clique}, we prove that every synchronous game $\cG=(I,O,\lambda)$ is hereditarily $*$-equivalent to an independence number game associated with the ``graph of the game" $X(\cG)$ arising from Atserias et al \cite{AMRSSV19}. In particular, for each $t \leq hered$, we prove that $\cG$ has a winning $t$-strategy if and only if $\alpha_t(X(\cG))=|I|$. It also follows that $\cG$ has a winning $t$-strategy if and only if $\omega_t(\overline{X(\cG)})=|I|$.

\section{Preliminaries}\label{section: preliminaries}

A \textbf{synchronous non-local game} $\cG$ is a triple $(I,O,\lambda)$, where $I$ is the finite question set, $O$ finite is the answer set, and $\lambda:O \times O \times I \times I \to \{0,1\}$ is the rule function, satisfying $\lambda(a,b,x,x)=0$ for all $x \in I$ and $a \neq b$ in $O$. Our main example in this paper is the \textbf{graph homomorphism game} $\text{Hom}(G,H)$, where $G$ and $H$ are finite, simple undirected graphs. The question set is $I=V(G)$; the answer set is $O=V(H)$; and the rules are
\begin{itemize}
\item $\lambda(a,b,x,x)=\delta_{ab}$ (synchronicity), and
\item $\lambda(a,b,x,y)=0$ if $(x,y)$ is an edge in $G$ but $(a,b)$ is not an edge in $H$. (adjacency)
\end{itemize}

The possible strategies for non-local games are given by quantum bipartite correlations, which are probability densities in a finite-input, finite-output system corresponding to two space-like separated players that may share entanglement. In general, given inputs $x$ and $y$ for Alice and Bob, respectively, there is an associated joint probability $p(a,b|x,y)$ that Alice and Bob output $a$ and $b$, respectively. This probability density arises from a strategy upon which the players agree beforehand; the players are not allowed to communicate once the game begins. We call the probability density $(p(a,b|x,y))$ a \textbf{winning strategy} for $\cG$ if, whenever $\lambda(a,b,x,y)=0$, we have $p(a,b|x,y)=0$.

There are essentially four models that are often analyzed for synchronous non-local games: the classical (or local) model, the quantum model, the quantum approximate model, and the quantum commuting model. These models are often abbreviated as loc, q, qa and qc, respectively. We refer the reader to \cite{J+,Fr11,JNVWY20,Oz13} and the references therein for more information on these correlations; in this paper, we will only focus on \textbf{synchronous correlations}; these are the correlations that satisfy $p(a,b|x,x)=0$ whenever $a \neq b$.

The work of \cite{PSSTW16}, along with \cite{KPS18} for the quantum approximate model, gives us the following criteria for synchronous correlations: any correlation $p(a,b|x,y)$ that is synchronous and belongs to the qc model can be written as $p(a,b|x,y)=\tau(E_{a,x}E_{b,y})$, where $\tau$ is a tracial state on a unital $C^*$-algebra $\cA$, and the elements $\{E_{a,x}\}_{a \in O, \, x \in I}$ are projections (self-adjoint idempotents) satisfying $\sum_{a \in O} E_{a,x}=1_{\cA}$ for each $x \in I$. The set of all such correlations is denoted by $C_{qc}^s(n,k)$, where $n=|I|$ and $k=|O|$.

The correlations in the qa model that are synchronous have a similar characterization, except that $\cA$ can be arranged to be the ultrapower $\cR^{\cU}$ of the hyperfinite $II_1$ factor (see \cite{BO08} for more on $\cR$). The set of such correlations is denoted by $C_{qa}^s(n,k)$.

For the q model, the algebra $\cA$ can be arranged to be finite-dimensional (in particular, a direct sum of matrix algebras); the set of such correlations is denoted by $C_q^s(n,k)$. For the loc model, the algebra $\cA$ can be arranged to be abelian (in particular, since it is already unital, one can arrange to have $\cA=C(X)$ for some compact Hausdorff space $X$). The set of all local synchronous correlations on $n$ inputs and $k$ outputs is denoted by $C_{loc}^s(n,k)$.

For $t \in \{loc,q,qa,qc\}$, a synchronous non-local game $\cG=(I,O,\lambda)$ with $|I|=n$ and $|O|=k$ is said to have a \textbf{winning $t$-strategy} if there exists a correlation $p=(p(a,b|x,y)) \in C_t^s(n,k)$ that is a winning strategy for $\cG$. Applying this terminology to the graph homomorphism game yields quantum versions of the chromatic number, independence number, and clique number. Indeed, for graphs $G$ and $H$, we write $G \to_t H$ if $\text{Hom}(G,H)$ has a winning $t$-strategy. Then, with $K_n$ denoting the complete graph on $n$ vertices and $\overline{G}$ denoting the graph complement of a graph $G$, we define
\begin{align*}
\chi_t(G)&=\min \{ c: G \to_t K_c\} \\
\alpha_t(G)&=\max \{ m: K_m \to_t \overline{G}\} \\
\omega_t(G)&=\max \{ m: K_m \to_t G\}
\end{align*}
as the $t$-chromatic number, the $t$-independence number, and the $t$-clique number, respectively, of $G$. It is well known that, if $t=loc$, then these quantities are precisely the usual chromatic, independence and clique numbers of $G$.

Associated to any synchronous non-local game $\cG=(I,O,\lambda)$ is the \textbf{game $*$-algebra} of $\cG$, denoted $\cA(\cG)$, which is the universal unital $*$-algebra generated by elements $e_{a,x}$, $a \in O$, $x \in I$, satisfying
\begin{itemize}
\item $e_{a,x}^2=e_{a,x}=e_{a,x}^*$;
\item $\sum_{a \in O} e_{a,x}=1$ for all $x \in I$;
\item $e_{a,x}e_{b,y}=0$ whenever $\lambda(a,b,x,y)=0$.
\end{itemize}

Equivalently, one takes the free algebra $\bC[\mathbb{F}(|I|,|O|)]$ generated by $|I|$ free unitaries, each of order $|O|$, and takes the quotient by the $*$-closed, two-sided ideal $\cI(\cG)$ generated by the elements of the form
\begin{itemize}
\item $e_{a,x}^2-e_{a,x}$, $a \in O, \, x \in I$;
\item $e_{a,x}^*-e_{a,x}$, $a \in O, \, x \in I$;
\item $1-\sum_{a \in O} e_{a,x}$, $x \in I$;
\item $e_{a,x}e_{b,y}$, $(a,b,x,y) \in \lambda^{-1}(\{0\})$.
\end{itemize}
(See \cite{HMPS19} for more information on the game algebra of $\cG$.) In a unital $*$-algebra $\cA$, we call an element $p \in \cA$ \textbf{positive}, and write $p \geq 0$, if $p=x_1^*x_1+\cdots+x_m^*x_m$ for some elements $x_1,...,x_m \in \cA$. We similarly write $p \leq q$ if $q-p$ is positive in $\cA$. With these notions in mind, associated to the game $\cG$, the \textbf{hereditary closure} of the ideal $\cI(\cG)$, denoted $\cI^h(\cG)$, is the smallest two-sided $*$-closed ideal in $\bC[\mathbb{F}(|I|,|O|)]$ that is hereditary; that is, if $0 \leq f \leq g$ and $g \in \cI^h(\cG)$, then $f \in \cI^h(\cG)$ as well. Then the \textbf{hereditary game $*$-algebra} is $\cA^h(\cG)=\mathbb{C}[\mathbb{F}(|I|,|O|)]/\cI^h(\cG)$. Note that, if $\cA^h(\cG) \neq (0)$, then whenever $x_1,...,x_m \in \cA^h(\cG)$ satisfy $x_1^*x_1+ \cdots + x_m^*x_m=0$, then $x_1=x_2=\cdots=x_m=0$. The set of positive elements in a hereditary $*$-algebra form a cone. Thus, the hereditary game algebra is a useful object when one is considering equivalences between synchronous games that require positivity arguments, as such arguments often cannot be used at the level of the game algebra.

Winning strategies for $\cG$ arise precisely from unital $*$-homomorphisms of $\cA(\cG)$. For convenience, given a unital $*$-algebra $\cB$, we will write $\cA(\cG) \to \cB$ if there is a unital $*$-homomorphism $\pi:\cA(\cG) \to \cB$. Then the following hold:
\begin{itemize}
\item $\cG$ has a winning loc strategy if and only if $\cA(\cG) \to \bC$.
\item $\cG$ has a winning $q$ strategy if and only if $\cA(\cG) \to M_d(\bC)$ for some $d \in \bN$.
\item $\cG$ has a winning $qa$ strategy if and only if $\cA(\cG) \to \cR^{\cU}$.
\item $\cG$ has a winning $qc$ strategy if and only if $\cA(\cG) \to (\cA,\tau)$ for some unital $C^*$-algebra $\cA$ with a (faithful) tracial state $\tau$.
\end{itemize}
More information on these statements can be found in \cite{PSSTW16,KPS18,HMPS19}. Based on these correspondences, three more general models have been considered for synchronous non-local games. Such a game $\cG$ is said to have a \textbf{winning $C^*$-strategy} if $\cA(\cG) \to \cB$ for some unital $C^*$-algebra $\cB$. We say that $\cG$ has a \textbf{winning hereditary strategy} (sometimes abbreviated ``hered") if $\cA(\cG) \to \cC$ for some hereditary unital $*$-algebra $\cC$. Equivalently, $\cG$ has a winning hereditary strategy if $\cA^h(\cG) \neq (0)$. Finally, we simply say that $\cA(\cG)$ has a \textbf{winning algebraic strategy} (often abbreviated ``alg") if the game algebra $\cA(\cG)$ is non-trivial; i.e., $\cA(\cG) \neq (0)$.

For convenience, we assign an ordering to the set of models described by \[loc \leq q \leq qa \leq qc \leq C^* \leq hered \leq alg.\] In particular, if $t_1,t_2 \in \{loc,q,qa,qc,C^*,hered,alg\}$ and $t_1 \leq t_2$, and if $\cG$ has a winning $t_1$ strategy, then it also has a winning $t_2$ strategy. In particular, in the case of the graph parameter games, if $t_1 \leq t_2$, then $\chi_{t_2}(G) \leq \chi_{t_1}(G)$, while $\alpha_{t_1}(G) \leq \alpha_{t_2}(G)$ and $\omega_{t_1}(G) \leq \omega_{t_2}(G)$.
 
One important tool for understanding classes of synchronous games is $*$-equivalence. There are a few kinds of equivalence in the literature, the first of which is $*$-equivalence. For our purposes it is also helpful to define $*$-subequivalence, which is new.

\begin{definition}
Let $\cG_1$ and $\cG_2$ be synchronous non-local games. We will say that $\cG_1$ is \textbf{$*$-subequivalent} to $\cG_2$ if there is a unital $*$-homomorphism $\cA(\cG_2) \to \cA(\cG_1)$. We will call $\cG_1$ and $\cG_2$ \textbf{$*$-equivalent} if $\cG_1$ is $*$-subequivalent to $\cG_2$ and $\cG_2$ is $*$-subequivalent to $\cG_1$.
\end{definition}

We note that the notion of $*$-equivalence here agrees with $*$-equivalence as it has appeared in the literature. These maps are only assumed to be unital $*$-homomorphisms, and may be neither injective nor surjective. Nevertheless, if $\cG_1$ and $\cG_2$ are $*$-equivalent, and if $t \in \{loc,q,qa,qc,C^*,hered,alg\}$, then $\cG_1$ has a winning $t$-strategy if and only if $\cG_2$ does. (See, for example, \cite{KPS18}.) On the other hand, if $\cG_1$ is only $*$-subequivalent to $\cG_2$, then whenever $\cG_1$ has a winning $t$-strategy, so does $\cG_2$.

A stronger notion of $*$-equivalence is when the game algebras are actually $*$-isomorphic. Several examples of this phenomenon are exhibited in \cite{H22}. Another type of equivalence (which first appeared in \cite{BCEHPSW20}) is hereditary $*$-equivalence. We say that two synchronous games $\cG_1$ and $\cG_2$ are \textbf{hereditarily $*$-equivalent} if there are unital $*$-homomorphisms $\pi:\cA^h(\cG_1) \to \cA^h(\cG_2)$ and $\rho:\cA^h(\cG_2) \to \cA^h(\cG_1)$.  The same type of equivalence of winning strategies holds for hereditarily $*$-equivalent games $\cG_1$ and $\cG_2$, except possibly for $t=alg$.

The last kind of equivalence that we will use appears to be new, although arguments in \cite{Ji13} essentially relied on such an equivalence.

\begin{definition}
Let $\cG_1$ and $\cG_2$ be synchronous games. We say that $\cG_1$ is \textbf{weakly $*$-subequivalent} to $\cG_2$ if, for each unital $*$-homomorphism $\varphi:\cA(\cG_1) \to \cD$ into a non-zero unital $*$-algebra $\cD$, there exists a non-zero subalgebra $\cC$ of $\cD$ with unit $1_{\cC}$ and a unital $*$-homomorphism $\rho:\cA(\cG_2) \to \cC$.

We will call the games $\cG_1$ and $\cG_2$ \textbf{weakly $*$-equivalent} if $\cG_1$ is weakly $*$-subequivalent to $\cG_2$ and $\cG_2$ is weakly $*$-subequivalent to $\cG_1$.
\end{definition}

We observe that, if $\cG_1$ is weakly $*$-subequivalent to $\cG_2$, then winning strategies for $\cG_1$ can be transformed into winning strategies for $\cG_2$.

\begin{proposition}
Suppose that $\cG_1$ and $\cG_2$ are synchronous games, and that $\cG_1$ is weakly $*$-subequivalent to $\cG_2$. If $t \in \{loc,q,qa,qc,C^*,hered,alg\}$ and $\cG_1$ has a winning $t$-strategy, then so does $\cG_2$.
\end{proposition}

\begin{proof}
If $\cA(\cG_1)=\cA(\cG_2)=(0)$, then there is nothing to prove, so we assume that $\cA(\cG_1) \neq (0)$ and that $\pi:\cA(\cG_1) \to \cD$ is a unital $*$-homomorphism, where $\cD \neq (0)$. By assumption, there is a non-zero subalgebra $\cC$ of $\cD$ with unit $1_{\cC}$, and a unital $*$-homomorphism $\rho:\cA(\cG_2) \to \cC$. Most of the cases for the different choices of $t \in \{loc,q,qa,qc,C^*,hered,alg\}$ will amount to similar proofs, though some have subtleties.

If $t=loc$, then we can arrange to have $\cD=\bC$. But the only non-zero subalgebra $\cC$ of $\cD$ with unit is $\bC$, so $\cC=\bC$. Hence, the homomorphism $\rho:\cA(\cG_2) \to \cC=\bC$ shows that $\cG_2$ has a winning loc strategy.

If $t=q$, then we can arrange to have $\cD$ to be a finite-dimensional $C^*$-algebra. Then clearly the non-zero subalgebra $\cC$ is also finite-dimensional, so $\cG_2$ has a winning $q$-strategy.

If $t=qa$, then we can arrange to have $\cD \subseteq \cR^{\cU}$, where $\cR^{\cU}$ is a tracial ultrapower of the hyperfinite $II_1$ factor $\cR$. Then $\cC$ is a subalgebra of $\cR^{\cU}$, with possibly a different unit $p=1_{\mathcal{C}}$. But then $\mathcal{C}$ is a subalgebra of $p\mathcal{R}^{\cU}p$, and since $p$ is non-zero, the corner algebra $p\mathcal{R}^{\mathcal{U}}p$ is isomorphic to $\mathcal{R}^{\cU}$ via some $*$-isomorphism $\theta$ \cite{B11}. Thus, $\mathcal{A}(\cG_2)$ has a unital $*$-homomorphism into the image of $\theta(\mathcal{C}) \subseteq \cR^{\cU}$, with this inclusion being unital. Thus, $\cG_2$ has a winning $qa$-strategy.

If $t=qc$, then $\cD$ can be arranged to be a unital $C^*$-algebra with faithful tracial state $\tau$. Choose a non-zero subalgebra $\cC$ with unit $1_{\cC}$ such that $\cA(\cG_1) \to \cC$. As $\tau$ is faithful, $\tau(1_{\cC})>0$, so the state $\tau_{\cC}(x)=\frac{\tau(x)}{\tau(1_{\cC})}$ is a faithful tracial state on $\cC$. This shows that $\cG_1$ has a winning $qc$ strategy.

The case of $t=C^*$ is trivial. Suppose $\cD$ is a hereditary unital $*$-algebra and $\cC$ is a non-zero subalgebra of $\cD$. If $x_1,...,x_n \in \mathcal{C}$ and $x_1^*x_1+\cdots+x_n^*x_n=0$, then since $\mathcal{D}$ is hereditary, $x_1=...=x_n=0$, so $\mathcal{C}$ is hereditary as well. Hence, the case $t=hered$ follows as well. The case $t=alg$ is immediate from special case of the definition of weak $*$-subequivalence for the identity map $\id:\cA(\cG_1) \to \cA(\cG_1) \neq (0)$.

\end{proof}

\section{Three-coloring game for $K_3 \times K_3$}\label{section: 3x3 rook's graph}

In this section, we will show that the three-coloring game for the $3 \times 3$ rook's graph has an abelian game algebra. This game algebra will be our key tool in the reduction theorem in the next section. The $3 \times 3$ rook's graph encodes the moves that a rook can make on a $3 \times 3$ chess board. Two vertices are connected by an edge if a rook can move between the two vertices in a single move on a chess board; see Figure \ref{figure: 3x3 rook's graph}. This is also the Cartesian graph product $K_3 \times K_3$ of the complete graph on $3$ vertices with itself. 

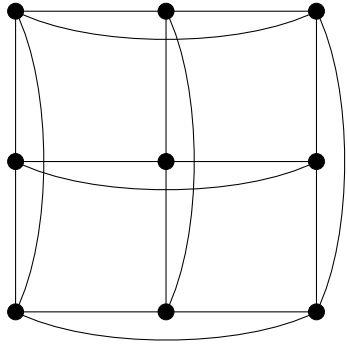
\begin{figure}[!ht]
\begin{tikzpicture}
\tikzset{vertex/.style = {shape=circle,draw,minimum size=1.5em}}
\tikzset{edge/.style = {->,> = latex'}}
\filldraw (0,4) circle (3pt);
\filldraw (0,2) circle (3pt);
\filldraw (0,0) circle (3pt);
\filldraw (-2,4) circle (3pt);
\filldraw (-2,2) circle (3pt);
\filldraw (-2,0) circle (3pt);
\filldraw (2,4) circle (3pt);
\filldraw (2,2) circle (3pt);
\filldraw (2,0) circle (3pt);
\draw (-2,4) -- (2,4);
\draw (-2,2) -- (2,2);
\draw (-2,0) -- (2,0);
\draw (-2,4) -- (-2,0);
\draw (0,4) -- (0,0);
\draw (2,4) -- (2,0);
\draw (-2,4) .. controls (-1,3.5) and (1,3.5) .. (2,4);
\draw (-2,2) .. controls (-1,1.5) and (1,1.5) .. (2,2);
\draw (-2,0) .. controls (-1,-0.5) and (1,-0.5) .. (2,0);
\draw (-2,4) .. controls (-1.5,3) and (-1.5,1) .. (-2,0);
\draw (0,4) .. controls (0.5,3) and (0.5,1) .. (0,0);
\draw (2,4) .. controls (2.5,3) and (2.5,1) .. (2,0);
\end{tikzpicture}
\caption{\small $K_3 \times K_3$, the $3 \times 3$ rook's graph.}
\label{figure: 3x3 rook's graph}
\end{figure}

The most useful representation of the graph is given with vertex set $\{ (i,j): 1 \leq i,j \leq 3\}$, with the adjacency relations $(i,j) \sim (k,\ell)$ if and only if exactly one of $i=k$ or $j=\ell$ holds. Our approach will also recover a result of Z. Ji \cite[Lemma~4]{Ji13}, which states that, in the three coloring game for a triangular prism, projections corresponding to non-adjacent vertices must commute. Our final proof of this result, though, is simpler than the one presented in \cite{Ji13}.

First, we collect a few results that are well-known, but make the proofs of the main result in this section simpler. The first is well-known, and the proof below essentially appeared in \cite{BES94}; however, one step is simplified by working in a $*$-algebra.

\begin{proposition}
\label{proposition: three projections summing to zero}
Suppose that $p_1,p_2,p_3$ are self-adjoint idempotents in a unital $*$-algebra $\mathcal{A}$. If $p_1+p_2+p_3=0$, then $p_1=p_2=p_3=0$.
\end{proposition}

\begin{proof}
Rearranging yields 
\begin{equation}
p_1+p_2=-p_3. \label{equation: two projections summing to negative}
\end{equation}
Squaring both sides and using the fact that $p_j^2=p_j$ for each $j$ gives 
\begin{equation}
p_1+p_1p_2+p_2p_1+p_2=p_3. \label{equation: square of two projections summing to negative}
\end{equation} 
Adding (\ref{equation: two projections summing to negative}) and (\ref{equation: square of two projections summing to negative}) gives
\begin{equation}
2p_1+p_1p_2+p_2p_1+2p_2=0. \label{equation: from three idempotents summing to zero}
\end{equation}
Pre-multiplying (\ref{equation: from three idempotents summing to zero}) by $p_1$ yields $2p_1+3p_1p_2+p_1p_2p_1=0$. Taking adjoints, we obtain $2p_1+3p_2p_1+p_1p_2p_1=0$. It follows that $p_1p_2=p_2p_1$. Using equation (\ref{equation: from three idempotents summing to zero}) we obtain $p_1+p_2+p_1p_2=0$. Post-multiplying by $p_1p_2$ and using commutativity, one arrives at the equation $p_1p_2=0$. Then $p_1+p_2=0$, and post-multiplying by $p_2$ gives $p_2=0$. Similarly, $p_1=0$, so that $p_3=0$ as well.
\end{proof}

\begin{proposition}
\label{proposition: three projections summing to 1}
If $p_1,p_2,p_3$ are self-adjoint idempotents in a unital $*$-algebra $\cA$ with $p_1+p_2+p_3=1$, then $p_ip_j=0$ for all $i \neq j$.
\end{proposition}

\begin{proof}
We have 
\begin{equation}
p_1+p_2=1-p_3, \label{equation: two projections summing to perp}
\end{equation} 
and squaring (\ref{equation: two projections summing to perp}) gives 
\begin{equation}
p_1+p_1p_2+p_2p_1+p_2=1-p_3. \label{equation: from three projections summing to 1}
\end{equation}
Subtracting (\ref{equation: two projections summing to perp}) from (\ref{equation: from three projections summing to 1}) yields $p_1p_2+p_2p_1=0$, so that $p_1p_2=-p_2p_1$. Post-multiplying by $p_1$ gives $p_1p_2p_1=-p_2p_1$, but the left side is self-adjoint, so we have $-p_2p_1=(-p_2p_1)^*=-p_1p_2$, so that $p_1$ and $p_2$ commute. Then since $p_1p_2=-p_2p_1=-p_1p_2$, we have $p_1p_2=0$. The other orthogonality relations are similar.
\end{proof}

As a result, we can show the following, which is important in the three-coloring context.

\begin{proposition}
\label{proposition: 3x3 matrix of projections with rows summing to 1}
Suppose that $Q=(p_{ij})_{i,j=1}^3$ is a $3 \times 3$ matrix with entries in a unital $*$-algebra $\mathcal{A}$. If each $p_{ij}$ is a self-adjoint idempotent, and if $p_{i1}+p_{i2}+p_{i3}=1$ for each $i=1,2,3$, and if $p_{ij}p_{kj}=0$ for all $i \neq k$ and $j=1,2,3$, then $p_{1j}+p_{2j}+p_{3j}=1$ for all $j=1,2,3$. In particular, $Q$ is a quantum permutation.
\end{proposition}

\begin{proof}
For each $j=1,2,3$, the column sum $q_j=p_{1j}+p_{2j}+p_{3j}$ is a self-adjoint idempotent in $\mathcal{A}$, since the entries in the sum are pairwise orthogonal. Since the row sums of $Q$ are all $1$, it follows that
\[ q_1+q_2+q_3=\sum_{i,j=1}^3 p_{ij}=\sum_{i=1}^3 (p_{i1}+p_{i2}+p_{i3})=3.\]
But each $1-q_j$ is a self-adjoint idempotent and $(1-q_1)+(1-q_2)+(1-q_3)=0$. By Proposition \ref{proposition: three projections summing to zero}, we must have $1-q_j=0$ for each $j$, so that $q_j=1$ for each $j$.
\end{proof}

The next proposition gives a common method of proving facts about $3 \times 3$ quantum permutations, and the $3$-coloring game for triangles, triangular prisms and the $3 \times 3$ rook's graph. The proof is inspired by, and a generalization of, the proof that the entries of a $3 \times 3$ quantum permutation must commute from \cite{LMR20}.

\begin{proposition}
\label{proposition: the key three projection proposition}
Let $\{e_1,e_2,e_3\}$ and $\{f_1,f_2,f_3\}$ be PVMs in a unital $*$-algebra $\cA$. Suppose that $[e_i,f_i]=0$ for all $i=1,2,3$. Then $[e_i,f_j]=0$ for all $1 \leq i,j \leq 3$.
\end{proposition}

\begin{proof}
By symmetry, it suffices to show that $e_1f_2=f_2e_1$. We note that
\begin{align*}
e_1f_2&=e_1f_2(e_1+e_2+e_3) \\
&=e_1f_2e_1+e_1f_2e_2+e_1f_2e_3.
\end{align*}
As $e_2$ and $f_2$ commute, the second term is equal to $e_1e_2f_2=0$. As $f_2=1-f_1-f_3$, we can rewrite
\begin{align*}
e_1f_2&=e_1f_2e_1+e_1(1-f_1-f_3)e_3 \\
&=e_1f_2e_1+e_1e_3-e_1f_1e_3-e_1f_3e_3 \\
&=e_1f_2e_1-e_1f_1e_3-e_1f_3e_3,
\end{align*}
where the last line follows since $e_1e_3=0$. Now, since $[e_i,f_i]=0$ and $e_ie_j=0$ for $i \neq j$, we have $e_1f_1e_3=0=e_1f_3e_3$. Thus, $e_1f_2=e_1f_2e_1$ is self-adjoint, so we have $e_1f_2=(e_1f_2)^*=f_2e_1$, and we are done.
\end{proof}

A special case of the above is the following proposition.

\begin{proposition}
\label{proposition: key three projection proposition, orthogonal version}
Let $\{p_1,p_2,p_3\}$ and $\{q_1,q_2,q_3\}$ be PVMs in a unital $*$-algebra $\cA$. If $p_iq_i=0$ for all $i=1,2,3$, then $p_iq_j=q_jp_i$ for all $i,j$. In particular, if two vertices are adjacent in a graph, then in the three coloring game algebra, any projections corresponding to those vertices commute with each other.
\end{proposition}

We also recover by Proposition \ref{proposition: the key three projection proposition} the well-known fact that entries of $3 \times 3$ quantum permutations must commute. The approach is similar to what is found in \cite{LMR20}.

\begin{proposition}
\label{proposition: 3x3 quantum permutation commutes}
Suppose that $\mathcal{A}$ is a unital $*$-algebra and that $p_{ij}$ are self-adjoint idempotents in $\mathcal{A}$ for $1 \leq i,j \leq 3$. If $\sum_{i=1}^3 p_{ij}=\sum_{j=1}^3 p_{ij}=1$ for all $i,j$, then $[p_{ij},p_{k\ell}]=0$ for all $i,j,k,\ell$.
\end{proposition}

\begin{proof}
By Proposition \ref{proposition: three projections summing to 1}, we have $p_{ij}p_{ik}=0$ for $j \neq k$ and $p_{ij}p_{kj}=0$ for $i \neq k$, so any two entries from a common row or a common column of $P=(p_{ij})_{i,j=1}^3$ commute. For the rest of the commutation relations, by symmetry, we need only check that $[p_{11},p_{22}]=0$. But this follows by Proposition \ref{proposition: key three projection proposition, orthogonal version} using the PVMs $\{ p_{11},p_{21},p_{31}\}$ and $\{p_{12},p_{22},p_{32}\}$ and the orthogonality relations.
\end{proof}

\begin{proposition}
\label{proposition: 3 coloring of a triangle}
If $G$ is a triangle with vertices $\{1,2,3\}$, then for each $i$, $e_{i1}+e_{i2}+e_{i3}=1$ in $\mathcal{A}(\text{Hom}(G,K_3))$. In particular, $\mathcal{A}(\text{Hom}(G,K_3))$ is $*$-isomorphic to the universal unital $*$-algebra generated by entries of a $3 \times 3$ quantum permutation, and is abelian.
\end{proposition}

\begin{proof}
The matrix $Q=(e_{ij})_{i,j=1}^3$ has every column sum equal to $1$, by the relations of the game algebra. Moreover, by the adjacency relations, $e_{ij}e_{ik}=0$ for $j \neq k$. By Proposition \ref{proposition: 3x3 matrix of projections with rows summing to 1} applied to the transpose of $Q$, it follows that $Q$ is a quantum permutation, so $e_{i1}+e_{i2}+e_{i3}=1$. The final claim is easy to verify.
\end{proof}

We will use the following lemma several times throughout.

\begin{lemma}
\label{lemma: qperms with ij entries orthogonal}
Suppose that $P=(p_{ij})_{i,j=1}^3$ and $Q=(q_{ij})_{i,j=1}^3$ be quantum permutations in a unital $*$-algebra $\cA$. If $p_{ij}q_{ij}=0$ for each $1 \leq i,j \leq 3$, then $p_{ij}q_{k\ell}=q_{k\ell}p_{ij}$ for all $i,j,k,\ell$.
\end{lemma}

\begin{proof}
First, use Proposition \ref{proposition: the key three projection proposition} with the PVMs $\{p_{1j},p_{2j},p_{3j}\}$ and $\{q_{1j},q_{2j},q_{3j}\}$ to obtain $[p_{ij},q_{kj}]=0$ for all $k=1,2,3$. Then use Proposition \ref{proposition: the key three projection proposition} with the PVMs $\{p_{i1},p_{i2},p_{i3}\}$ and $\{q_{k1},q_{k2},q_{k3}\}$ to get $[p_{ij},q_{k\ell}]=0$.
\end{proof}

We now obtain a simpler proof of commutativity of projections corresponding to non-adjacent vertices in the $3$-coloring game for the triangular prism. For simplicity, we label the projections in the $3$-coloring game for the triangular prism using the vertex labels in Figure \ref{figure: triangular prism}. For example, the three projections corresponding to vertex $p$ are written as $p_1$, $p_2$ and $p_3$.

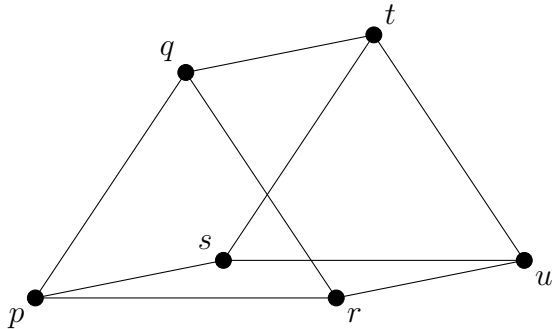
\begin{figure}[!ht]
\begin{tikzpicture}
\tikzset{vertex/.style = {shape=circle,draw,minimum size=1.5em}}
\tikzset{edge/.style = {->,> = latex'}}
\filldraw (-2,0) circle (3pt);
\filldraw (2,0) circle (3pt);
\filldraw (0,3) circle (3pt);
\filldraw (0.5,0.5) circle (3pt);
\filldraw (4.5,0.5) circle (3pt);
\filldraw (2.5,3.5) circle (3pt);
\node[below left] (P) at (-2,0) {$p$};
\node[above left] (Q) at (0,3) {$q$};
\node[below right] (R) at (2,0) {$r$};
\node[above left] (S) at (0.5,0.5) {$s$};
\node[below right] (T) at (4.5,0.5) {$u$};
\node[above right] (U) at (2.5,3.5) {$t$};
\draw (-2,0) -- (0,3) -- (2.5,3.5) -- (0.5,0.5) -- (-2,0);
\draw (-2,0) -- (2,0) -- (0,3);
\draw (2,0) -- (4.5,0.5) -- (2.5,3.5);
\draw (4.5,0.5) -- (0.5,0.5);
\end{tikzpicture}
\caption{Triangular prism}
\label{figure: triangular prism}
\end{figure}

\begin{proposition}
\label{proposition: three coloring of triangular prism}
Let $G$ be a triangular prism given as in Figure \ref{figure: triangular prism}. Then in $\cA(\text{Hom}(G,K_3))$, projections corresponding to any two non-adjacent vertices commute. In particular, $\cA(\text{Hom}(G,K_3))$ is abelian.
\end{proposition}

\begin{proof}
By symmetry of the graph, to show that projections corresponding to non-adjacent vertices commute, we need only show that $[p_i,t_j]=0$ for all $1 \leq i,j \leq 3$. For a fixed $i \in \{1,2,3\}$, we note that, by Proposition \ref{proposition: 3 coloring of a triangle}, we have $p_i+q_i+r_i=1=s_i+t_i+u_i$. Moreover, $p_is_i=q_it_i=r_iu_i=0$, so by Proposition \ref{proposition: the key three projection proposition}, $[p_i,t_i]=0$ for each $i$. To show the rest of the proof, we use Proposition \ref{proposition: the key three projection proposition} again with the PVMs $\{ p_1,p_2,p_3\}$ and $\{t_1,t_2,t_3\}$, since $[p_i,t_i]=0$ for each $i=1,2,3$. The fact that $\cA(\text{Hom}(G,K_3))$ is abelian follows since projections corresponding to adjacent vertices already commute in the $3$-coloring game, by Proposition \ref{proposition: key three projection proposition, orthogonal version}.
\end{proof}

We now arrive at the analogue of Proposition \ref{proposition: 3x3 quantum permutation commutes} for what we call ``$3 \times 3 \times 3$" quantum permutations, which is interesting in its own right.

\begin{theorem}
\label{theorem: 3x3x3 quantum permutation}
Suppose that $p_{ijk}$ are self-adjoint idempotents in a unital $*$-algebra $\mathcal{A}$, for $i,j,k \in \{1,2,3\}$, satisfying
\[ \sum_{i=1}^3 p_{ijk}=\sum_{j=1}^3 p_{ijk}=\sum_{k=1}^3 p_{ijk}=1.\]
Then $[p_{ijk},p_{abc}]=0$ for all $i,j,k,a,b,c \in \{1,2,3\}$.
\end{theorem}

\begin{proof}
First, suppose that at least one of $i=a$, $j=b$ or $k=c$ holds. By symmetry of the relations given, we may assume that $i=a$. As $P_i=(p_{ijk})_{j,k=1}^3$ is a $3 \times 3$ quantum permutation, we have $[p_{ijk},p_{abc}]=0$ by Proposition \ref{proposition: 3x3 quantum permutation commutes}. Thus, we need only show that $[p_{ijk},p_{abc}]=0$ whenever $i \neq a$, $j \neq b$ and $k \neq c$. For this, we consider the quantum permutations $P_i=(p_{ijk})_{j,k=1}^3$ and $P_a=(p_{abc})_{b,c=1}^3$. By the assumed summation relations and Proposition \ref{proposition: three projections summing to 1}, we have $p_{ijk}p_{ajk}=0$ for each $j,k$ since $i \neq a$. An application of Lemma \ref{lemma: qperms with ij entries orthogonal} shows that $[p_{ijk},p_{abc}]=0$, as desired.
\end{proof}

\begin{corollary}
\label{corollary: 3 coloring of 3x3 rook is abelian}
$\mathcal{A}(\text{Hom}(K_3 \times K_3,K_3))$ is abelian.
\end{corollary}

\begin{proof}
Write the vertices of $K_3 \times K_3$ as $(i,j)$ for $1 \leq i,j \leq 3$; note that the adjacency relations are $(i,j) \sim (k,\ell)$ if and only if exactly one of $i=k$ or $j=\ell$ holds. Let $e_{c,(i,j)}$ be the self-adjoint idempotent generator of the game algebra corresponding to the color $c$ for the vertex $(i,j)$, for $1 \leq c,i,j \leq 3$. By definition of the game algebra, we have $e_{1,(i,j)}+e_{2,(i,j)}+e_{3,(i,j)}=1$ for all $i,j$. As the vertex set $\{ (1,j),(2,j),(3,j)\}$ forms a triangle, by Proposition \ref{proposition: 3 coloring of a triangle}, we have $e_{c,(1,j)}+e_{c,(2,j)}+e_{c,(3,j)}=1$ for each $c,j$. Similarly, since $\{ (i,1),(i,2),(i,3)\}$ forms a triangle, we have $e_{c,(i,1)}+e_{c,(i,2)}+e_{c,(i,3)}=1$ for all $c,j$. It follows by Theorem \ref{theorem: 3x3x3 quantum permutation} that $e_{c,(i,j)}$ commutes with $e_{d,(k,\ell)}$ for all $1 \leq c,d,i,j,k,\ell \leq 3$. As the game algebra is generated by the projections $e_{c,(i,j)}$, it must be abelian.
\end{proof}

\section{Weak $*$-equivalence to three-coloring games}\label{section: equivalence}

In this section, we will prove that every synchronous game $\cG=(I,O,\lambda)$ is weakly $*$-equivalent to a three-coloring game for a graph associated with an asymmetric version $\lambda_{\text{asym}}$ of the rule function $\lambda$ of $\cG$. In particular, we will prove that there is an associated graph $G_{\lambda_{\text{asym}}}$ for which $\text{Hom}(G_{\lambda_{\text{asym}}},K_3)$ is weakly $*$-subequivalent to $\cG$ (Theorem \ref{theorem: coloring game weakly subequivalent to original game}) and $\cG$ is $*$-subequivalent to $\text{Hom}(G_{\lambda_{\text{asym}}},K_3)$ (Theorem \ref{theorem: game subequivalent to coloring}).

\begin{definition}
Let $\cG=(I,O,\lambda)$ be a synchronous non-local game. An \textbf{asymmetric rule function} for $\cG$ is a function $\lambda_{\text{asym}}: O \times O \times I \times I \to \{0,1\}$ with the requirements that 
\begin{itemize}
\item
$\lambda_{\text{asym}}(a,b,x,x)=\delta_{ab}$;
\item $\lambda_{\text{asym}}(a,b,x,y)+\lambda_{\text{asym}}(b,a,y,x) \geq 1$ for all $x \neq y$; and
\item $\lambda_{\text{asym}}(a,b,x,y)\lambda_{\text{asym}}(b,a,y,x)=\lambda(a,b,x,y)\lambda(b,a,y,x)$ for all $a,b,x,y$.
\end{itemize}
\end{definition}

If $\cG=(I,O,\lambda)$ is a synchronous game with $|I| \geq 2$, then one example of an asymmetric rule function for $\cG$ is the function
\[ \lambda_{\text{asym}}(a,b,x,y)=\begin{cases} \delta_{ab} & x=y \\
\lambda(a,b,x,y)\lambda(b,a,y,x) & x<y \\
1 & x>y. \end{cases}\]

\begin{proposition}
\label{proposition: asymmetric rule function}
Let $\cG=(I,O,\lambda)$ be a synchronous non-local game with $|I| \geq 2$ and $\lambda(a,a,x,x)=1$ for all $a \in O$ and $x \in I$. If $\lambda_{\text{asym}}$ is an asymmetric rule function for $\cG$, then the game $*$-algebras $\cA((I,O,\lambda))$ and $\cA((I,O,\lambda_{\text{asym}}))$ are $*$-isomorphic.
\end{proposition}

\begin{proof}
Write the generators of $\cA((I,O,\lambda))$ as $e_{a,x}$ and the generators of $\cA((I,O,\lambda_{\text{asym}})$ as $f_{a,x}$ for $a \in O$ and $x \in I$. The first condition guarantees that $\lambda_{\text{asym}}(a,b,x,x)=\delta_{ab}$. If $x \neq y$ and $\lambda_{\text{asym}}(a,b,x,y)=0$, then $f_{a,x}f_{b,y}=0$, while the equation
\begin{equation}
0=\lambda_{\text{asym}}(a,b,x,y)\lambda_{\text{asym}}(b,a,y,x)=\lambda(a,b,x,y)\lambda(b,a,y,x)
\end{equation}
implies that either $\lambda(a,b,x,y)=0$ or $\lambda(b,a,y,x)=0$. Thus, either $e_{a,x}e_{b,y}=0$ or $e_{b,y}e_{a,x}=0$. But these conditions are equivalent by taking adjoints. Since $\sum_{a \in O} e_{a,x}=1$ and $\sum_{a \in O} f_{a,x}=1$, all the relations defining $\cA((I,O,\lambda_{\text{asym}}))$ are satisfied by the generators of $\cA((I,O,\lambda))$, so the map $\pi:\cA((I,O,\lambda_{\text{asym}})) \to \cA((I,O,\lambda))$ given on generators by $\pi(f_{a,x})=e_{a,x}$ for all $a,x$ extends to a unital $*$-homomorphism.

Similarly, if $\lambda(a,b,x,y)=0$ and $x \neq y$, then the last condition on $\lambda_{\text{asym}}$ implies that $\lambda_{\text{asym}}(a,b,x,y)=\lambda_{\text{asym}}(b,a,y,x)=0$, giving $f_{a,x}f_{b,y}=0$. This shows that the map $\rho:\cA((I,O,\lambda)) \to \cA((I,O,\lambda_{\text{asym}}))$ given by $\rho(e_{a,x})=f_{a,x}$ for all $a,x$ extends to a unital $*$-homomorphism. Clearly $\pi$ and $\rho$ are inverses of each other, so the game algebras are $*$-isomorphic.
\end{proof}

\begin{remark}
\label{remark: lambda(a,a,x,x)=1}
If $|I| \geq 2$, then we can assume that $\lambda(a,a,x,x)=1$ for all $a \in O$ and $x \in I$. Indeed, if $\lambda(a,a,x,x)=0$, then in the game algebra, we will have $e_{a,x}=e_{a,x}^2=0$. If there is a $y \in I$ with $x \neq y$, then in the game algebra, one has
\[ 0=e_{a,x}=e_{a,x}\sum_{b \in O} e_{b,y}=\sum_{b \in O} e_{a,x}e_{b,y}.\]
As $e_{b,y}e_{b',y}=0$ for $b \neq b'$, post-multiplying by a fixed $e_{b,y}$ yields $e_{a,x}e_{b,y}=0$ for all $b \in O$. Thus, one can re-define the rule function for $\cG$ by replacing the condition that $\lambda(a,a,x,x)=0$ with the condition that $\lambda(a,b,x,y)=0$ for some $y \in I \setminus \{x\}$ and all $b \in O$. The resulting game has game algebra that is $*$-isomorphic to $\cA(\cG)$.

If the game $\cG$ has $|I|=1$, then one can add in a second question where only the first answer is allowed, and this will yield the same game $*$-algebra as the original game. Similarly, if $\cG$ is a synchronous game and not all of the questions have the same number of possible answers, then we can add in additional answers and force them to be disallowed by the rule function. It follows by the above paragraph that we can always assume without loss of generality that $\cG$ is a synchronous game with $|I| \geq 2$, $|O| \geq 3$, and $\lambda(a,a,x,x)=1$ for all $a \in O$ and $x \in I$.
\end{remark}

In this way, applying Proposition \ref{proposition: asymmetric rule function} to this modified rule function for $\cG$, we may assume henceforth that $\cG=(I,O,\lambda)$ is a synchronous non-local game with $|I| \geq 2$ and with $\lambda$ being asymmetric (which implies that $\lambda(a,a,x,x)=1$ for all $a \in O$ and $x \in I$). The use of an asymmetric rule function in this section is mainly cosmetic--our aim is to eliminate unnecessary rules that are automatically enforced in the game algebra. As a result, our resulting graph corresponding to $\cG$ will have less vertices than if $\cG$ had a rule function that was not asymmetric.

Given a synchronous non-local game $\cG=(I,O,\lambda)$ with $|I| \geq 2$ and $\lambda$ asymmetric, the graph $G_{\lambda}$ will be constructed from the non-local game $\cG$ as follows. For simplicity, we write $I=\{1,...,n\}$ and $O=\{1,...,k\}$. We start with a triangle $\Delta$ with vertices $A$, $B$ and $C$. Next, for each $1 \leq \alpha \leq k-2$ and $1 \leq x \leq n$, we add a copy $R_{\alpha,x}$ of $K_3 \times K_3$, with vertices written as $\{ v(i,j,\alpha,x): 1 \leq i,j \leq 3\}$, with the usual adjacency relations that $v(i,j,\alpha,x) \sim v(i',j',\alpha,x)$ if and only if exactly one of $i=i'$ or $j=j'$ holds. We make the identifications
\begin{equation} v(1,2,\alpha,x)=B \text{ for all } 1 \leq x \leq n, \, 1 \leq \alpha \leq k-2 \label{1,2 identification} 
\end{equation} 
(That is, each $R_{\alpha,x}$ has $B$ as its $(1,2)$-vertex.) We also impose the identifications and relations 
\begin{align}
v(3,2,\alpha,x)&=v(1,1,\alpha+1,x), \, \text{ for all } 1 \leq \alpha \leq k-3, \, 1 \leq x \leq n, \label{3,2,1,1 identification} \\
A &\sim v(3,3,\alpha,x) \text{ for all } 1 \leq \alpha \leq k-2, \, 1 \leq x \leq n, \label{base zeros 3,3 entries} \\
C &\sim v(2,1,\alpha,x) \text{ for all } 1 \leq \alpha \leq k-2, \, 1 \leq x \leq n. \label{2,1 third color zero}
\end{align}
To ensure that this part of the graph encodes $n$ PVMs with $k$ outputs each, we add in a triangular prism $T_{\alpha,x}$, for each $1 \leq \alpha \leq k-2$ and $1 \leq x \leq n$. The triangular prism $T_{\alpha,x}$ will have one triangle given by the subgraph of $R_{\alpha,x}$ with vertices $\{ v(1,1,\alpha,x),B,v(1,3,\alpha,x)\}$, and the other triangle given by the vertices $\{ t(1,\alpha,x),A,t(2,\alpha,x)\}$, with $v(1,1,\alpha,x) \sim t(1,\alpha,x)$, $B \sim A$ and $v(1,3,\alpha,x) \sim t(2,\alpha,x)$.

There are certain vertices in the subgraphs $R_{\alpha,x}$ that are of utmost importance to the proofs of the main theorems in this section, so we reserve special notation for those vertices. For each $1 \leq x \leq n$ and $1 \leq a \leq k$, we define
\begin{equation}
\widehat{v}(a,x)=\begin{cases} v(1,1,1,x) & \text{if } a=1 \\ v(2,1,a-1,x) & \text{if } 2 \leq a \leq k-1 \\ v(2,2,k-2,x) & \text{if } a=k. \end{cases} \label{specialvertices}
\end{equation}
To encode the orthogonality relations from the rule function, we first define the sets
\[ \mathcal{E}_{\lambda}=\{ (a,b,x,y) \in \lambda^{-1}(\{0\}): (a,b) \in \{(1,1),(1,k),(k,1),(k,k), \, x \neq y\}\]
and
\[ \mathcal{F}_{\lambda}=\{(a,b,x,y) \in \lambda^{-1}(\{0\}):2 \leq a,b \leq k-1, \, x \neq y\}.\]
By Remark \ref{remark: lambda(a,a,x,x)=1}, the only tuples $(a,b,x,x)$ in $\lambda^{-1}(\{0\})$ are those with $a \neq b$, and such orthogonality will already be enforced in the subgraphs $R_{\alpha,x}$. To enforce the remaining orthogonality relations, there are two possible cases.

\textbf{Case 1.} If $(a,b,x,y) \in \lambda^{-1}(\{0\}) \setminus (\cE_{\lambda} \cup \cF_{\lambda})$, then we add in the adjacency relation
\begin{equation}
\widehat{v}(a,x) \sim \widehat{v}(b,y). \label{orthogonality adjacency}
\end{equation}

\textbf{Case 2.} If $(a,b,x,y) \in \cE_{\lambda} \cup \cF_{\lambda}$, then we construct a copy $Q_{a,b,x,y}$ of $K_3 \times K_3$, denoting the vertices by $q(i,j,a,b,x,y)$, $1 \leq i,j \leq 3$ with the usual adjacency relations, along with the identifications
\begin{align}
q(1,1,a,b,x,y)&=\widehat{v}(a,x) \label{orthogonality rooks 1}\\
q(2,2,a,b,x,y)&=\widehat{v}(b,y)  \label{orthogonality rooks 2}\\
A &\sim q(3,3,a,b,x,y) \text{ for each } (a,b,x,y) \in \cE_{\lambda} \cup \cF_{\lambda}. \label{base zeros 3,3 entries orthos}\\
q(1,2,a,b,x,y)&=\begin{cases} B & (a,b,x,y) \in \cE_{\lambda} \\ C & (a,b,x,y) \in \cF_{\lambda}. \end{cases} \label{1,2 entry of orthogonality rooks}
\end{align}

The resulting graph obtained from the subgraphs of the form $\Delta$, $R_{\alpha,x}$, $T_{\alpha,x}$ and $Q_{a,b,x,y}$, along with the identifications and relations described in (\ref{1,2 identification})--(\ref{1,2 entry of orthogonality rooks}), will be denoted by $G_{\lambda}$.

Before we prove that $\text{Hom}(G_{\lambda},K_3)$ and $\cG$ are weakly $*$-equivalent, we first need a lemma regarding the center of $\cA(\text{Hom}(G_{\lambda},K_3))$.

\begin{lemma}
\label{lemma: center of 3-coloring algebra}
For each $c=1,2,3$, the projections $e_{c,A}$, $e_{c,B}$ and $e_{c,C}$ belong to the center of $\cA(\text{Hom}(G_{\lambda},K_3))$.
\end{lemma}

\begin{proof}
Since $\{A,B,C\}$ is a triangle, by Proposition \ref{proposition: 3 coloring of a triangle}, $e_{c,A}+e_{c,B}+e_{c,C}=1$ for each $1 \leq c \leq 3$. Thus, it suffices to show that, for each vertex $\nu$ in $G_{\lambda}$ and for each $1 \leq c,d \leq 3$, at least two of the projections $e_{c,A}$, $e_{c,B}$ and $e_{c,C}$ commute with $e_{d,\nu}$.

Based on the identification (\ref{1,2 identification}) and Corollary \ref{corollary: 3 coloring of 3x3 rook is abelian}, $e_{c,B}$ commutes with $e_{d,v(i,j,\alpha,x)}$ for all $1 \leq d,i,j \leq 3$, $1 \leq \alpha \leq k-2$ and $1 \leq x \leq n$. Since $A \sim v(3,3,\alpha,x)$ by adjacency relation (\ref{base zeros 3,3 entries orthos}), by Proposition \ref{proposition: key three projection proposition, orthogonal version} we have $[e_{c,A},e_{d,v(3,3,\alpha,x)}]=0$ for all $c,d,\alpha,x$. Similarly, by adjacency relation (\ref{2,1 third color zero}), $[e_{c,C},e_{d,v(2,1,\alpha,x)}]=0$ for all $c,d,\alpha,x$. As $e_{c,A}=1-e_{c,B}-e_{c,C}$, it follows that $[e_{c,A},e_{d,v(2,1,\alpha,x)}]=0$ as well. Using the triangular prism $T_{\alpha,x}$ and Proposition \ref{proposition: three coloring of triangular prism} guarantees that $e_{c,A}$ commutes with $e_{d,v(1,1,\alpha,x)}$ and $e_{d,v(1,3,\alpha,x)}$ for all $d,\alpha,x$. Thus, $e_{c,A}$ commutes with projections corresponding to the vertices $v(i,j,\alpha,x)$ for $(i,j) \in \{(1,1),(1,2),(1,3),(2,1),(3,3)\}$. Since $\sum_i e_{c,v(i,j,\alpha,x)}=\sum_j e_{c,v(i,j,\alpha,x)}=1$ for each $i,j$, it follows that $[e_{c,A},e_{d,v(i,j,\alpha,x)}]=0$ for all $i,j$. Therefore, each of $e_{c,A}$, $e_{c,B}$ and $e_{c,C}$ commute with $e_{d,v(i,j,\alpha,x)}$ for all $d,i,j,\alpha,x$. We note that each of $e_{c,A}$, $e_{c,B}$ and $e_{c,C}$ automatically commute with $\{e_{d,t(1,\alpha,x)},e_{d,t(2,\alpha,x)}\}_{d,\alpha,x}$ by Proposition \ref{proposition: three coloring of triangular prism}, since $A$ and $B$ belong to the triangular prism $T_{\alpha,x}$.

The argument for the subgraphs $Q_{a,b,x,y}$ is a bit simpler. If $(a,b,x,y) \in \cE_{\lambda} \cup \cF_{\lambda}$, then by identification (\ref{1,2 entry of orthogonality rooks}), for all $(a,b,x,y) \in \cE_{\lambda} \cup \cF_{\lambda}$ and $1 \leq d,i,j \leq 3$, at least one of $e_{c,B}$ or $e_{c,C}$ commutes with $e_{d,q(i,j,a,b,x,y)}$, so we will be done if we show that $e_{c,A}$ also commutes with each $e_{d,q(i,j,a,b,x,y)}$. By relation (\ref{base zeros 3,3 entries orthos}), $[e_{c,A},e_{d,q(3,3,a,b,x,y)}]=0$. As $q(1,1,a,b,x,y)$ and $q(2,2,a,b,x,y)$ already arose in previous subgraphs, we have $[e_{c,A},e_{d,q(i,j,a,b,x,y)}]=0$ for all $(i,j) \in \{(1,1),(1,2),(2,2),(3,3)\}$. A similar argument to the one for the subgraph $R_{\alpha,x}$ shows that $[e_{c,A},e_{d,q(i,j,a,b,x,y)}]=0$ for all possible indices. The result follows.
\end{proof}

\begin{theorem}
\label{theorem: coloring game weakly subequivalent to original game}
If $\pi:\cA(\text{Hom}(G_{\lambda},K_3)) \to \mathcal{D}$ is a non-zero unital $*$-homomorphism, then there exists a non-zero subalgebra $\cC$ of $\cD$ with unit $1_{\cC}$ and a unital $*$-homomorphism $\pi:\cA(\cG) \to \cC$. In particular, $\text{Hom}(G_{\lambda},K_3)$ is weakly $*$-subequivalent to $\cG$.
\end{theorem}

\begin{proof}
Suppose that $\pi:\mathcal{A}(\text{Hom}(G_{\lambda},K_3)) \to \cD$ is a non-zero unital $*$-homomorphism into the unital $*$-algebra $\cD$. By replacing $\cD$ with the $*$-algebra generated by the range of $\pi$, we may assume that $\pi$ is surjective. Since $\{A,B,C\}$ is a triangle, by Proposition \ref{proposition: 3 coloring of a triangle} we can write
\[ 1=\sum_{i,j,k=1}^3 e_{i,A}e_{j,B}e_{k,C}=\sum_{\substack{1 \leq i,j,k \leq 3 \\ i \neq j, \, j \neq k, \, i \neq k}} e_{i,A}e_{j,B}e_{k,C}.\]
At least one of these terms must have non-zero image in $\cD$, and all of them belong to the center of $\cA(\text{Hom}(G_{\lambda},K_3))$ by Lemma \ref{lemma: center of 3-coloring algebra}. By re-labelling the colors if necessary, we may assume that $\pi(e_{1,A}e_{2,B}e_{3,C}) \neq 0$. We define $\cC=\pi(e_{1,A}e_{2,B}e_{3,C})\cD$, which is a non-zero subalgebra of $\cD$ with unit $1_{\cC}=\pi(e_{1,A}e_{2,B}e_{3,C})$. For each $1 \leq c \leq 3$ and vertex $\nu$ in $G_{\lambda}$, we let 
\[p_{c,\nu}=1_{\cC} \pi(e_{c,\nu})=\pi(e_{1,A}e_{2,B}e_{3,C})\pi(e_{c,\nu}), \, 1 \leq c \leq 3, \, \nu \in V(G_{\lambda}).\] The orthogonality relations arising from the triangle $\{A,B,C\}$ immediately imply that $p_{1,A}=p_{2,B}=p_{3,C}=1_{\cC}$, and all other projections in $\cC$ corresponding to $A,B,C$ are zero. The identification (\ref{1,2 identification}) and the adjacency relation (\ref{base zeros 3,3 entries}) force $p_{1,v(1,2,\alpha,x)}=p_{1,v(3,3,\alpha,x)}=0$ for all $\alpha,x$. Define $g_{a,x}=p_{1,\widehat{v}(a,x)}$; we will show that the projections $\{ g_{a,x}: 1 \leq a \leq k, \, 1 \leq x \leq n\}$ constitute a representation of $\cA(\cG)$.

For the subgraph $R_{1,x}$ and color $1$, recalling that $v(1,2,1,x)=B$, we note the form of the quantum permutation \begin{equation}
(p_{1,v(i,j,1,x)})_{i,j=1}^3=\begin{pmatrix} g_{1,x} & 0 & 1-g_{1,x} \\ g_{2,x} & p_{1,v(2,2,1,x)} & g_{1,x} \\ p_{1,v(3,1,1,x)} & p_{1,v(3,2,1,x)} & 0 \end{pmatrix}, \label{color 1 of R_{1,x}}
\end{equation}
so that $g_{1,x}+g_{2,x}=1-p_{1,v(3,1,1,x)}=p_{1,v(3,2,1,x)}$; moreover, $g_{1,x}g_{2,x}=0$. 
Working inductively, assume that $g_{1,x}+\cdots+g_{a,x}=p_{1,v(3,2,a-1,x)}$ where $2 \leq a \leq k-3$ and $g_{\alpha,x}g_{\beta,x}=0$ for all $1 \leq \alpha<\beta \leq a$; we will show that $g_{1,x}+\cdots+g_{a+1,x}=p_{1,v(3,2,a,x)}$ and that $g_{\ell,x}g_{a+1,x}=0$ for all $1 \leq \ell \leq a$. We notice that, by identification (\ref{3,2,1,1 identification}), we have
\[ p_{1,v(1,1,a,x)}=p_{1,v(3,2,a-1,x)}=g_{1,x}+\cdots+g_{a,x}.\]
Looking at color $1$ for the subgraph $R_{a+1,x}$, the quantum permutation $(p_{1,v(i,j,a,x)})_{i,j=1}^3$ is of the form
\begin{equation} (p_{1,v(i,j,a,x)})_{i,j=1}^3=\begin{pmatrix} g_{1,x}+\cdots+g_{a,x} & 0 & 1-(g_{1,x}+\cdots+g_{a,x}) \\ g_{a+1,x} & p_{1,v(2,2,a,x)} & g_{1,x}+\cdots+g_{a,x} \\ p_{1,v(3,1,a,x)} & p_{1,v(3,2,a,x)} & 0 \end{pmatrix}. \label{color 1 R_{a+1,x}}
\end{equation}
By considering the sum of column 1 and the sum of row 3 from (\ref{color 1 R_{a+1,x}}), we have
\[ g_{1,x}+\cdots+g_{a+1,x}=1-p_{1,v(3,1,a,x)}=p_{1,v(3,2,a,x)},\]
establishing the first part of the claim. For the other part of the claim, notice that the elements $g_{1,x}+\cdots+g_{a,x}$, $g_{a+1,x}$ and $g_{1,x}+\cdots+g_{a+1,x}$ are all self-adjoint idempotents. Whenever $p,q,r$ are self-adjoint idempotents with $p+q=r$, we have $p+q+(1-r)=1$, forcing $pq=0$ by Proposition \ref{proposition: three projections summing to 1}. In particular,
\begin{equation} (g_{1,x}+\cdots+g_{a,x})g_{a+1,x}=0. \label{sum of a projections times a+1st}
\end{equation}
By assumption, $g_{\alpha,x}g_{\beta,x}=0$ whenever $1 \leq \alpha,\beta \leq a$ with $\alpha \neq \beta$. Pre-multiplying (\ref{sum of a projections times a+1st}) by $g_{\beta,x}$ yields $g_{\beta,x}g_{a+1,x}=0$ for all $1 \leq \beta \leq a$. Therefore, it follows by induction that $g_{1,x}+\cdots+g_{k-2,x}=p_{1,v(3,2,k-3,x)}=p_{1,v(1,1,k-2,x)}$ and $g_{a,x}g_{b,x}=0$ for all $1 \leq a,b \leq k-2$ with $a \neq b$. To show that $\{ g_{a,x}\}_{a=1}^k$ is a PVM, we consider the form of the quantum permutation arising from color $1$ for the subgraph $R_{k-2,x}$, which is
\begin{equation}
(p_{1,v(i,j,k-2,x)})_{i,j=1}^3=\begin{pmatrix} g_{1,x}+\cdots+g_{k-2,x} & 0 & 1-(g_{1,x}+\cdots+g_{k-2,x}) \\ g_{k-1,x} & g_{k,x} & g_{1,x}+\cdots+g_{k-2,x} \\ p_{1,v(3,1,k-2,x)} & 1-g_{k,x} & 0 \end{pmatrix}. \label{color 1 R_{k-2,x}}
\end{equation}
The sum on row 2 in (\ref{color 1 R_{k-2,x}}) shows that $\displaystyle \sum_{a=1}^k g_{a,x}=1$. Applying Proposition \ref{proposition: three projections summing to 1} to the projections $g_{1,x}+\cdots+g_{k-2,x}$, $g_{k-1,x}$ and $g_{k,x}$, a similar argument demonstrates that $g_{a,x}g_{b,x}=0$ for all $1 \leq a,b \leq k$ with $a \neq b$. It follows that $\{g_{1,x},...,g_{k,x}\}$ is a PVM for each $x$.

It remains to show that $g_{a,x}g_{b,y}=0$ whenever $\lambda(a,b,x,y)=0$. If $(a,b,x,y) \in \lambda^{-1}(\{0\}) \setminus (\cE_{\lambda} \cup \cF_{\lambda})$, then the adjacency relation (\ref{orthogonality adjacency}) forces $g_{a,x}g_{b,y}=0$, since $g_{a,x}=p_{1,\widehat{v}(a,x)}$ and $g_{b,y}=p_{1,\widehat{v}(b,y)}$.

If $(a,b,x,y) \in \cE_{\lambda} \cup \cF_{\lambda}$, then the quantum permutation for the subgraph $Q_{a,b,x,y}$ corresponding to color $1$ is of the form
\begin{equation}
(p_{1,q(i,j,a,b,x,y)})_{i,j=1}^3= \begin{pmatrix} g_{a,x} & 0 & 1-g_{a,x} \\ p_{1,q(2,1,a,b,x,y)} & g_{b,y} & p_{1,q(2,3,a,b,x,y)} \\ g_{b,y} & 1-g_{b,y} & 0 \end{pmatrix}, \label{color 1 Q(a,b,x,y)}
\end{equation}
using adjacency relation (\ref{base zeros 3,3 entries orthos}) and identification (\ref{1,2 entry of orthogonality rooks}). The first column of (\ref{color 1 Q(a,b,x,y)}) forces $g_{a,x}g_{b,y}=0$ for all $(a,b,x,y) \in \cE_{\lambda} \cup \cF_{\lambda}$.

Finally, using the universal property of $\cA(\cG)$, there is a unital $*$-homomorphism $\pi:\cA(\cG) \to \cC$ given by $\pi(f_{a,x})=g_{a,x}$ for all $1 \leq a \leq k$ and $1 \leq x \leq n$, as desired.
\end{proof}

\begin{theorem}
\label{theorem: game subequivalent to coloring}
There is a unital $*$-homomorphism $\rho:\cA(\text{Hom}(G_{\lambda},K_3)) \to \cA(\cG)$. In particular, $\cG$ is $*$-subequivalent to $\text{Hom}(G_{\lambda},K_3)$.
\end{theorem}

\begin{proof}
Write the generators of $\cA(\cG)$ as $f_{a,x}$, for $1 \leq a \leq k$ and $1 \leq x \leq n$. We will exhibit an algebraic $3$-coloring of $G_{\lambda}$, by defining algebraic colorings on each of the parts of $G_{\lambda}$ and verifying that all the orthogonality conditions and identifications hold. For convenience, we will set $f_{[a,b],x}=f_{a,x}+f_{a+1,x}+\cdots+f_{b,x}$ for $1 \leq a \leq b \leq k$, and $f_{[a,b],x}=0$ if $a>b$. The triangle $\{A,B,C\}$ is colored by the assignments $A \mapsto (1,0,0)$, $B \mapsto (0,1,0)$ and $C \mapsto (0,0,1)$. For each subgraph $R_{\alpha,x}$ of $G_{\lambda}$ for $1 \leq \alpha \leq k-2$ and $1 \leq x \leq n$, we use the algebraic coloring given by the three quantum permutations $H_{c,\alpha,x}=(h_{c,v(i,j,\alpha,x)})_{i,j=1}^3$ for colors $c \in \{1,2,3\}$ given by
\[ H_{1,\alpha,x}=\begin{pmatrix} f_{[1,\alpha],x} & 0 & f_{[\alpha+1,k],x} \\ f_{\alpha+1,x} & f_{[\alpha+2,k],x} & f_{[1,\alpha],x} \\ f_{[\alpha+2,k],x} & f_{[1,\alpha+1],x} & 0 \end{pmatrix},\]
\[ H_{2,\alpha,x}=\begin{pmatrix} 0 & 1 & 0 \\ 1-f_{\alpha+1,x} & 0 & f_{\alpha+1,x} \\ f_{\alpha+1,x} & 0 & 1-f_{\alpha+1,x} \end{pmatrix}, \]
\[ H_{3,\alpha,x}=\begin{pmatrix} f_{[\alpha+1,k],x} & 0 & f_{[1,\alpha],x} \\ 0 & f_{[1,\alpha+1],x} & f_{[\alpha+2,k],x} \\ f_{[1,\alpha],x} & f_{[\alpha+2,k],x} & f_{\alpha+1,x} \end{pmatrix}.\]
We note that $\displaystyle\sum_{c=1}^3 h_{c,v(i,j,\alpha,x)}=\sum_{i=1}^3 h_{c,v(i,j,\alpha,x)}=\sum_{j=1}^3 h_{c,v(i,j,\alpha,x)}=1$. The identifications and adjacency relations involving $A,B,C$ also hold. Moreover, we note that $h_{c,v(3,2,\alpha,x)}=h_{c,v(1,1,\alpha+1,x)}$ for all $1 \leq \alpha \leq k-3$, so (\ref{1,2 identification})-(\ref{2,1 third color zero}) hold. Thus, these assignments yield a valid algebraic $3$-coloring of $R_{\alpha,x}$, for each $1 \leq \alpha \leq k-2$ and $1 \leq x \leq n$.

For the triangular prism $T_{\alpha,x}$, the first triangle is already colored, via the assignment \[v(1,1,\alpha,x) \mapsto (f_{[1,\alpha],x},0,f_{[\alpha+1,k],x}), \, B \mapsto (0,1,0), \, v(1,3,\alpha,x) \mapsto (f_{[\alpha+1,k],x},0,f_{[1,\alpha],x}).\] The second triangle is colored using the assignments \[t(1,\alpha,x) \mapsto (0,f_{[\alpha+1,k],x},f_{[1,\alpha],x}), \,  A \mapsto (1,0,0), \,  t(2,\alpha,x) \mapsto (0,f_{[1,\alpha],x},f_{[\alpha+1,k],x}).\]

Next, we color each subgraph $Q_{a,b,x,y}$ for $(a,b,x,y) \in \cE_{\lambda} \cup \cF_{\lambda}$. If $(a,b,x,y) \in \cE_{\lambda}$, then the colorings of $\widehat{v}(a,x)$ and $\widehat{v}(b,y)$, based on the form of each $H_{c,\alpha,x}$ above, are $(f_{a,x},0,1-f_{a,x})$ and $(f_{b,y},0,1-f_{b,y})$, respectively. Since $f_{a,x}f_{b,y}=0$ and $q(1,2,a,b,x,y)=B$ by identification (\ref{1,2 entry of orthogonality rooks}), we can extend these assignments to a coloring of $Q_{a,b,x,y}$ using the three quantum permutations
\[ J_{1,a,b,x,y}=\begin{pmatrix} f_{a,x} & 0 & 1-f_{a,x} \\ 1-f_{a,x}-f_{b,y} & f_{b,y} & f_{a,x} \\ f_{b,y} & 1-f_{b,y} & 0 \end{pmatrix} \text{ (color 1)},\]
\[ J_{2,a,b,x,y}=\begin{pmatrix} 0 & 1 & 0 \\ f_{a,x}+f_{b,y} & 0 & 1-f_{a,x}-f_{b,y} \\ 1-f_{a,x}-f_{b,y} & 0 & f_{a,x}+f_{b,y} \end{pmatrix} \text{ (color 2)},\]
\[ J_{3,a,b,x,y}=\begin{pmatrix} 1-f_{a,x} & 0 & f_{a,x} \\ 0 & 1-f_{b,y} & f_{b,y} \\ f_{a,x} & f_{b,y} & 1-f_{a,x}-f_{b,y}\end{pmatrix} \text{ (color 3)}.\]
If $(a,b,x,y) \in \cF_{\lambda}$, then by identification (\ref{specialvertices}), the colorings of $\widehat{v}(a,x)$ and $\widehat{v}(b,y)$ are of the form $(f_{a,x},1-f_{a,x},0)$ and $(f_{b,y},1-f_{b,y},0)$, respectively. Similar to the last case, we can extend these assignments to a coloring of $Q_{a,b,x,y}$ using the three quantum permutations as above, but with colors $2$ and $3$ swapped, since $q(1,2,a,b,x,y)=C$ by identification (\ref{1,2 entry of orthogonality rooks}). These assignments will satisfy all the rules corresponding to each $Q_{a,b,x,y}$, while preserving the colorings of the vertices $B$, $C$ and $\{\widehat{v}(a,x): 1 \leq a \leq k, \, 1 \leq x \leq n\}$. Finally, if $(a,b,x,y) \in \lambda^{-1}(\{0\}) \setminus (\cE_{\lambda} \cup \cF_{\lambda})$, then exactly one of $a,b \in \{1,k\}$ and the other belongs to $\{2,...,k-1\}$. Without loss of generality, we assume that $a \in \{1,k\}$ and $b \in \{2,...,k-1\}$. Then $\widehat{v}(a,x)$ is either $v(1,1,1,x)$ or $v(2,2,k-2,x)$ by identification (\ref{specialvertices}), while $\widehat{v}(b,y)=v(2,1,b-1,y)$. In either case, $\widehat{v}(a,x) \mapsto (f_{a,x},0,1-f_{a,x})$ and $\widehat{v}(b,y) \mapsto (f_{b,y},1-f_{b,y},0)$. Thus, the products of projections corresponding to the same color for $\widehat{v}(a,x)$ and $\widehat{v}(b,y)$ are always zero when $(a,b,x,y) \in \lambda^{-1}(\{0\}) \setminus (\cE_{\lambda} \cup \cF_{\lambda})$, so the orthogonality relation arising from adjacency relation (\ref{orthogonality adjacency}) holds.

It follows that these assignments yield an algebraic $3$-coloring of $G_{\lambda}$, so there is a unital $*$-homomorphism $\rho:\cA(\text{Hom}(G_{\lambda},K_3)) \to \cA(\cG)$.
\end{proof}

Since $G_{\lambda}$ contains a triangle, we automatically have $\chi_{alg}(G_{\lambda}) \geq 3$ \cite{HMPS19}. On the other hand, every graph is algebraically $4$-colorable by \cite{HMPS19}, so $\chi_{alg}(G_{\lambda}) \in \{3,4\}$. The weak $*$-equivalence between $\cG$ and $\text{Hom}(G_{\lambda},K_3)$ yields the following corollary.

\begin{corollary}
\label{corollary: coloring equivalence}
Let $\cG=(I,O,\lambda)$ be a synchronous non-local game with $|O| \geq 3$; let $t \in \{loc,q,qa,qc,C^*,hered,alg\}$; and let $G_{\lambda}$ be the graph associated with $\cG$.
\begin{itemize}
\item[(1)] If $\cG$ has a winning $t$-strategy, then $\chi_t(G_{\lambda})=3$.
\item[(2)] If $\cG$ does not have a winning $t$-strategy, then $\chi_t(G_{\lambda}) \geq 4$.
\end{itemize}
Moreover, (2) becomes equality if $t=alg$.
\end{corollary}

\begin{proof}
By Theorem \ref{theorem: game subequivalent to coloring}, $\cG$ is $*$-subequivalent to $\text{Hom}(G_{\lambda},K_3)$. Thus, if $\cG$ has a winning $t$-strategy, then so does $\text{Hom}(G_{\lambda},K_3)$, yielding $\chi_t(G_{\lambda}) \leq 3$. Since $\chi_{alg}(G_{\lambda}) \geq 3$, we obtain equality.

Conversely, if $\cG$ does not have a winning $t$-strategy, then since $\text{Hom}(G_{\lambda},K_3)$ is weakly $*$-subequivalent to $\cG$ by Theorem \ref{theorem: coloring game weakly subequivalent to original game}, $\text{Hom}(G_{\lambda},K_3)$ cannot have a winning $t$-strategy either. Thus, $\chi_t(G_{\lambda})>3$, and it follows that $\chi_t(G_{\lambda}) \geq 4$. The claim for $t=alg$ is immediate, since every graph can be algebraically $4$-colored \cite{HMPS19}.
\end{proof}

\begin{remark}
Suppose that $\cG=(I,O,\lambda)$ is a synchronous non-local game with $|I|=n$ and $|O|=k$, and with $\lambda$ asymmetric. The triangle $\{A,B,C\}$ contributes $3$ vertices to $G_{\lambda}$. Each subgraph $R_{\alpha,x}$ contributes $7$ new vertices, with $R_{1,x}$ contributing an extra vertex, using identification (\ref{3,2,1,1 identification}). Each triangular prism $T_{a,x}$ contributes $2$ new vertices, for $1 \leq \alpha \leq k-2$ and $1 \leq x \leq n$. Lastly, there are $6$ new vertices for each subgraph $Q_{a,b,x,y}$, for each $(a,b,x,y) \in \cE_{\lambda} \cup \cF_{\lambda}$. Thus, the number of vertices in $G_{\lambda}$ is equal to
\[ 3+n+9n(k-2)+6|\cE_{\lambda}|+6|\cF_{\lambda}|.\]
Depending on what the original game $\cG$ is, it may be difficult to determine the sizes of $\cE_{\lambda}$ and $\cF_{\lambda}$. That being said, since $\cE_{\lambda} \cup \cF_{\lambda} \subseteq \lambda^{-1}\{0\}$ and $\cE_{\lambda} \cap \cF_{\lambda}=\emptyset$, we always have the upper bound
\[ |V(G_{\lambda})| \leq 3+n+9n(k-2)+6|\lambda^{-1}(\{0\})|.\]
\end{remark}

An important special case is a quantum version of L. Lov\'{a}sz's reduction theorem of the $k$-coloring problem to the $3$-coloring problem in polynomial time \cite{Lo73}.

\begin{corollary}
Let $G$ be a graph on $n$ vertices and $m$ edges. Let $k>3$ and let $t \in \{loc,q,qa,qc,C^*,hered,alg\}$. The $k$-coloring problem for $G$ in model $t$ is equivalent to the $3$-coloring problem in model $t$, for a certain graph $G_{\lambda}$ on $3+n+9n(k-2)+6mk$ vertices.
\end{corollary}

\begin{proof}
In the $k$-coloring game for $G$, the rules (aside from the synchronicity rules) are all of the form $\lambda(a,a,x,y)=0$ whenever $(x,y)$ is an edge in $G$. Thus, $(a,b,x,y) \in \cE_{\lambda}$ if and only if $x \sim y$ and $(a,b) \in \{ (1,1),(k,k)\}$, while $(a,b,x,y) \in \cF_{\lambda}$ if and only if $x \sim y$ and $2 \leq a=b \leq k-1$. Thus, $|\cE_{\lambda}|+|\cF_{\lambda}|=mk$, so the graph $G_{\lambda}$ from the main theorem has $3+n+9n(k-2)+6mk$ vertices, as desired.
\end{proof}

In work of Z. Ji \cite{Ji13}, synchronous non-local games, and more generally, binary constraint systems, were converted (in the loc and q models) to $3$-coloring games for graphs by considering a binary constraint system as a 3-SAT problem, and converting each clause in the satisfiability problem to a $3$-coloring of a subgraph of a large graph. This construction gives a $3$-colorable graph if and only if the original binary constraint system was satisfiable. To extend this to the $q$ model, one adds in triangular prisms to the main gadget graph as necessary to force variables from the same clause to commute with each other, which is a feature of quantum solutions to binary constraint systems \cite{Ji13}.

For a single clause $x_1 \vee x_2 \vee x_3$, the graph used in \cite{Ji13} has a ``control" triangle, much like our work here, which is fixed and does not depend on the particular clause. In addition, Ji constructs six vertices (one for each variable and one for its negation), and uses six other vertices (whose colorings depend on the assignments to $x_1,x_2,x_3$). For a synchronous non-local game with $n$ inputs and $k$ outputs, each PVM corresponds to the satisfiability problem $x_1+\cdots+x_k=1$, where each $x_j \in \{0,1\}$. In this setting, we are treating $1$ as ``True", while $0$ is treated as ``False". The negation of the variable $x_j$ is given by $\overline{x}_j$. The constraint $x_1+\cdots+x_k=1$ can be transformed into a 3-SAT problem using intermediate variables $r_{[1,j]}$ for $j=1,...,k-2$, where $r_{[1,1]}=x_1$ and $\overline{r}_{[1,k-1]}=\overline{x}_k$, and forcing $r_{[1,j]}+x_{j+1}+\overline{r}_{[1,j+1]}=1$. Each of these equations, essentially treated as 1-in-3 SAT problems (a 3-SAT problem where exactly one variable is allowed to be ``True"), can be converted to a 3-SAT via the expression \[(r_{[1,j]} \vee x_{j+1} \vee \overline{r}_{[1,j+1]}) \wedge (r_{[1,j]} \vee \overline{x}_{j+1}) \wedge(x_{j+1} \vee r_{[1,j+1]}) \wedge (r_{[1,j]} \vee r_{[1,j+1]}).\] 
Using Ji's construction in \cite{Ji13}, each element $x_j$ of the PVM, along with its complement, would appear as vertices, along with each $r_{[1,j]}$, $2 \leq j \leq k-2$, and its complement. Thus, one would need $2(k+k-3)=4k-6$ vertices for each PVM corresponding to the variables, plus $6$ vertices corresponding to the portion of the gadget for each clause that is not the control triangle and not the variables. In summary, one would need $4k-6+6(k-2)=10(k-2)+2$ vertices. For each rule of the game, including the synchronous rules, an additional six vertices are required, yielding a total of $3+2n+10n(k-2)+6|\lambda^{-1}(\{0\})|$ vertices.

However, for the quantum model, many more vertices are required. This problem arises since triangular prisms are necessary for variables in the same clause to commute. Due to the lack of triangles in the graph in \cite{Ji13}, many more intermediate vertices are required in each triangular prism constructed. One can reduce the number of vertices by only forcing two of the control triangle vertices to have projections that commute with all other projections in the gadget, but even enforcing this condition can require up to eight triangular prisms for each control vertex, with at least four intermediate vertices needed for each prism.

In our work, only one triangular prism is required for each pair $(\alpha,x)$ for $1 \leq x \leq n$ and $1 \leq \alpha \leq k-2$, and only two intermediate vertices are required per triangular prism. Hence, our construction provides a similar size graph as Ji's construction for the local model, but provides a significantly smaller graph in the $q$ model.

\begin{example}
\label{example: mermin-peres}
The Mermin-Peres magic square game \cite{Mer90} is the synchronous binary constraint system game where the set of equations is
\begin{align*}
x_1+x_2+x_3&=0, \\
x_4+x_5+x_6&=0, \\
x_7+x_8+x_9&=0, \\
x_1+x_4+x_7&=0, \\
x_2+x_5+x_8&=0, \\
x_3+x_6+x_9&=1.
\end{align*}
The question set is the set of equations. The answer set for each equation is the set of solutions in $\mathbb{Z}_2^3$ to the equation, which has $4$ elements. Hence, this game is a synchronous game with $n=6$ questions and $k=4$ answers. This game has a winning quantum strategy, but no winning classical strategy \cite{Mer90}. Using the main theorem, one obtains a graph $G$ with $\chi_q(G)=3<\chi_{loc}(G)$. After removing the synchronicity rules and making the rest of the rule function asymmetric, one finds that the number of disallowed $4$-tuples $(a,b,x,y)$, $x<y$, for this game is $72$. Thus, the graph $G$ has at most
\[ 3+6+9(6)(4-2)+6(72)=549\]
vertices. However, using the sets $\cE_{\lambda}$ and $\cF_{\lambda}$ instead of all of $\lambda^{-1}(\{0\})$ can significantly reduce the size of the graph. Indeed, if we label the questions (that is, the equations) of the game by the numbers $1$ through $6$, then we can label the $4$ possible solutions for each equation by $1,2,3,4$. For the first five equations, we use the assignments $(0,0,0) \to 1$, $(0,1,1) \mapsto 2$, $(1,0,1) \mapsto 3$ and $(1,1,0) \mapsto 4$, while for the final equation $x_3+x_6+x_9=1$, we use the assignments $(1,1,1) \mapsto 1$, $(1,0,0) \mapsto 2$, $(0,1,0) \mapsto 3$ and $(0,0,1) \mapsto 4$. (This is essentially the same idea as what is used in \cite{AMRSSV19} when transforming this game to a graph isomorphism game.) Using the structure of $\cE_{\lambda}$ and $\cF_{\lambda}$, one finds all $4$-tuples $(a,b,x,y)$, with $a,b \in \{1,2,3,4\}$ and $x,y \in \{1,...,6\}$ such that $x<y$, $\lambda(a,b,x,y)=0$, and \[(a,b) \in \{(1,4),(4,1),(1,1),(4,4),(2,3),(3,2),(2,2),(3,3)\}.\] 
A tedious calculation shows that there are $40$ such $4$-tuples, so the graph $G$ can be arranged to have
\[ 3+6+9(6)(4-2)+6(40)=357\]
vertices.
\end{example}

\begin{example}
\label{example: 338 vertices}
One of the main results of \cite{HMPS19} is that $\chi_{alg}(K_5)=4$ (and more generally, that \textit{every} graph can be algebraically $4$-colored). Note that $K_5$ has $5$ vertices and $10$ edges. Using the corollary, the graph $G_{\lambda}$ we obtain has
\[ 3+5+9(5)(4-2)+6(10)(4)=338\]
vertices. For this graph $G_{\lambda}$, we have $\chi_{alg}(G_{\lambda})=3<\chi_{hered}(G_{\lambda})$.
\end{example}

There are also graphs $G$ and $H$ with $\chi_{qa}(G)=3<\chi_q(G)$ and $\chi_{qc}(H)=3<\chi_{qa}(H)$. The graph $G$ arises from a synchronous game with a winning $qa$ strategy, but no winning $q$ strategy \cite{KPS18}, while the graph $H$ arises from a synchronous game with a winning $qc$ strategy, but no winning $qa$ strategy \cite{JNVWY20}. One could work through Slofstra's construction in \cite{Slo19} of a linear BCS game arising from a hyperlinear, non-residually finite group to obtain the number of vertices in $G$, although we have not done this here, since the linear system involves $184$ equations and $235$ variables. The game arising from \cite{JNVWY20} is not explicit, and as a result, the graph $H$ is not explicit either. That being said, this does show that $\chi_q$, $\chi_{qa}$ and $\chi_{qc}$ are all distinct chromatic numbers, which was not previously known. In fact, more is true:

\begin{corollary}
\label{corollary: chromatic numbers distinct}
The quantities $\chi_{loc}$, $\chi_q$, $\chi_{qa}$, $\chi_{qc}$, $\chi_{C^*}$, and $\chi_{alg}$ are all distinct.
\end{corollary}

\begin{proof}
Example \ref{example: 338 vertices} exhibits a graph $G$ with $\chi_{alg}(G)=3<\chi_{hered}(G) \leq \chi_{C^*}(G)$, so $\chi_{alg}$ and $\chi_{C^*}$ are distinct. It remains to show that $\chi_{qc}$ and $\chi_{C^*}$ are distinct. But this claim is achieved using a synchronous game of Paddock and Slofstra \cite{PS21} that has a winning $C^*$-strategy, but no winning $qc$-strategy.
\end{proof}

\begin{corollary}
\label{corollary: undecidable}
For $t \in \{q,qa,qc,C^*,hered,alg\}$, it is undecidable to determine whether a graph $G$ satisfies $\chi_t(G) \leq 3$.
\end{corollary}

\begin{proof}
Theorems of Slofstra show that, for $t \in \{q,qa,qc\}$, it is undecidable to determine whether (non-synchronous) linear system games have a winning $t$-strategy; see \cite{Slo19} for the cases $t \in \{q,qa\}$ and \cite{Slo20} for the case $t=qc$. The synchronous version of this game, $\text{syncBCS}(A,b)$, has a winning $t$-strategy for $t \in \{q,qa,qc\}$, if and only if the non-synchronous version has a winning $t$-strategy \cite{KPS18}. Using those results in conjunction with Corollary \ref{corollary: coloring equivalence} yields the desired result for $t \in \{q,qa,qc\}$.

For the other cases, we make particular use of the fact that the game $\cG$ for $t=qc$ can be arranged to be the synchronous version $\text{syncBCS}(A,b)$ of a linear system game. That game, in turn, is $*$-equivalent to the graph isomorphism game $\text{Iso}(G_{A,b},G_{A,0})$ \cite{BCEHPSW20,G21}. For the graph isomorphism game, however, the existence of a winning $qc$ strategy, a winning $C^*$ strategy, a winning hereditary strategy and a winning algebraic strategy are all equivalent \cite{BCEHPSW20}. So the undecidability passes to the graph isomorphism game for each of $t \in \{qc,C^*,hered,alg\}$, as all of these decision problems reduce to determining whether a syncBCS game has a winning $qc$ strategy. Applying Corollary \ref{corollary: coloring equivalence} yields the undecidability of determining whether $\chi_t(G) \leq 3$ for $t \in \{C^*,hered,alg\}$.
\end{proof}

Part of the utility of transforming synchronous games into graph coloring games is that, when graph coloring have winning strategies in any of the tracial models (i.e. for $t \in \{loc,q,qa,qc\}$), they can be won in a way that an honest verifier cannot gain any information about the strategy used. The key mathematical point here is the following proposition.

\begin{proposition}
\label{proposition: perfect zero knowledge}
Suppose that $G$ is a graph. Let $k \geq 2$ and $t \in \{loc,q,qa,qc\}$. If the $k$-coloring game $\text{Hom}(G,K_k)$ has a winning $t$-strategy, then there exists a winning $t$-strategy for $\text{Hom}(G,K_k)$ satisfying
\begin{align}
p(a,b|x,x)&=\begin{cases} 0 & a \neq b \\ \frac{1}{k} & a=b, \end{cases} \label{equation: perfect zero knowledge synchronous} \\
p(a,b|x,y)&=\begin{cases} 0 & x \sim y \text{ and } a=b \\ \frac{1}{k(k-1)} & x \sim y \text{ and } a\neq b.\end{cases} \label{equation: perfect zero knowledge adjacent}
\end{align}
\end{proposition}

\begin{proof}
By assumption, there is a tracial $C^*$-algebra $(\cA,\tau)$ and projection-valued measures $\{E_{a,x}\}_{a=1}^k$ in $\cA$, for each $x \in V(G)$, such that $E_{a,x}E_{b,y}=0$ whenever $x \sim y$ and $a=b$. We define projections $F_{a,x} \in M_{k!}(\cA)$ by
\[ F_{a,x}=\bigoplus_{\sigma \in S_k} E_{\sigma(a),x},\]
where $S_k$ denotes the group of permutations on the set $\{1,2,...,k\}$. Define $\rho=\tr_{k!} \otimes \tau$, where $\tr_{k!}$ is the normalized trace on $M_{k!}$. Then $\rho$ is a trace on $M_{k!}(\cA)$, and each $F_{a,x}$ is a projection in $M_{k!}(\cA)$. Moreover,
\[ \sum_{a=1}^k F_{a,x}=\bigoplus_{\sigma \in S_k} \left( \sum_{a=1}^k E_{\sigma(a),x}\right)=\bigoplus_{\sigma \in S_k} \left(\sum_{b=1}^k E_{b,x}\right)=I_{k!} \otimes 1,\]
so $\{F_{a,x}\}_{a=1}^k$ is a PVM for each $x \in V(G)$. Since $E_{a,x}E_{a,y}=0$ whenever $x \sim y$, it follows that $E_{\sigma(a),x}E_{\sigma(a),y}=0$ whenever $x\sim y$ and $\sigma \in S_k$. Thus, $F_{a,x}F_{a,y}=0$ if $x \sim y$, so that the probability density $p(a,b|x,y)=\rho(F_{a,x}F_{b,y})=\frac{1}{k!} \sum_{\sigma \in S_k} \tau(E_{\sigma(a),x}E_{\sigma(b),y})$ defines a winning $qc$-strategy for $\text{Hom}(G,K_k)$. 
If $t=loc$, then we can assume that $[E_{a,x},E_{b,y}]=0$ for all $a,b,x,y$, and clearly we obtain $[F_{a,x},F_{b,y}]=0$ as well. If $t=q$, then we can arrange for $\cA$ to be a matrix algebra $M_d$. Then the strategy $(p(a,b|x,y))$ arises from $M_{k!}(M_d)=M_{k!d}$, and hence is a winning $q$-strategy for the game.
If $t=qa$, then we can arrange for $\cA$ to be $\cR^{\cU}$. By uniqueness of the weakly separable hyperfinite $II_1$ factor $\cR$ \cite{Co76} and the ultrapower construction, one has the isomorphisms $M_{k!}(\cR^{\cU}) \simeq (M_{k!}(\cR))^{\cU} \simeq \cR^{\cU}$, so the strategy $(p(a,b|x,y))$ is a winning $qa$-strategy. Thus, for any $t \in \{loc,q,qa,qc\}$, the new strategy involving $\{ \{F_{a,x}\}_{a,x},\rho\}$ is a winning $t$-strategy for $\text{Hom}(G,K_k)$ so long as the original strategy $\{ \{E_{a,x}\}_{a,x},\tau\}$ is.

It remains to check that equations (\ref{equation: perfect zero knowledge synchronous}) and (\ref{equation: perfect zero knowledge adjacent}) are satisfied by $(p(a,b|x,y))$. The fact that $p(a,b|x,x)=0$ for $a \neq b$ is immediate since $p$ is synchronous. If $a=b$, then
\[ p(a,a|x,x)=\rho(F_{a,x}F_{a,x})=\rho(F_{a,x})=\frac{1}{k!} \sum_{\sigma \in S_k} \tau(E_{\sigma(a),x}).\]
Given $b \in \{1,...,k\}$, the number of $\sigma \in S_k$ satisfying $\sigma(a)=b$ is $(k-1)!$, so we have
\[ p(a,a|x,x)=\frac{1}{k!} \sum_{b=1}^k (k-1)! \tau(E_{b,x})=\frac{1}{k} \sum_{b=1}^k \tau(E_{b,x})=\frac{1}{k}.\]

The fact that $p(a,a|x,y)=0$ for $x \sim y$ is also immediate by the rules of the game. If $a \neq b$ and $x \sim y$, then
\[ p(a,b|x,y)=\rho(F_{a,x}F_{b,y})=\frac{1}{k!} \sum_{\sigma \in S_k} \tau(E_{\sigma(a),x}E_{\sigma(b),y}).\]
Given a pair $(c,d) \in \{1,...,k\}$ with $c \neq d$, the number of $\sigma \in S_k$ satisfying $\sigma(a)=c$ and $\sigma(b)=d$ is $(k-2)!$, so a similar argument shows that
\[ p(a,b|x,y)=\frac{1}{k!} \sum_{\substack{1 \leq c,d \leq k \\ c \neq d}} (k-2)! \tau(E_{c,x}E_{d,y})=\frac{1}{k(k-1)} \sum_{c,d} \tau(E_{c,x}E_{d,y})=\frac{1}{k(k-1)},\]
where the last two equalities follow since $E_{c,x}E_{c,y}=0$ for all $c$, and $\sum_{c,d} E_{c,x}E_{d,y}=\left(\sum_c E_{c,x}\right)^2=1$.
\end{proof}

An honest verifier is a referee that only asks question pairs $(x,y)$ where there is the possibility of the players losing--that is, where $\lambda(a,b,x,y)=0$ for some pair $(a,b)$. Proposition \ref{proposition: perfect zero knowledge} shows that, for graph coloring games with winning strategies, there are always winning strategies such that the probabilities involved with correct answers give no information to the referee. Thus, if the referee only asks question pairs $(x,y)$ for which $\lambda(a,b,x,y)=0$ for some pair $(a,b)$, then the referee does not gain insight into what kind of strategy the players may have used. Hence, graph coloring games exhibit perfect zero knowledge for an honest verifier.

Combining Proposition \ref{proposition: perfect zero knowledge} with Corollary \ref{corollary: coloring equivalence}, if $\cG=(I,O,\lambda)$ and if $G_{\lambda}$ is the graph associated with $\cG$, then there is a winning $t$-strategy $\cG$ if and only if there is a winning $t$-strategy for $\text{Hom}(G_{\lambda},K_3)$ that exhibits perfect zero knowledge for an honest verifier. Thus, if the players are asked to win a non-local game with very few winning $t$-strategies, then they can instead try to win the associated $3$-coloring game with a strategy that gives no information to an honest verifier. For example, the Mermin-Peres magic square game is a self-test: all winning $q$-strategies for the game are, up to a dilation of the Hilbert space, arising from $M_2 \otimes M_2$ and the unique trace on $M_2 \otimes M_2$ \cite{Fr22}. On the other hand, the players can instead exhibit a winning $q$-strategy for the $3$-coloring game for the associated graph, which reveals no information to an honest verifier, except that the players can $3$-color the graph in the $q$ model. This shows the potential use of the game transformation in this section, and also that this equivalence of synchronous games to $3$-coloring games does not preserve certain properties about the set of winning strategies. 

We close this section by answering an open problem posed by Helton, Meyer, Paulsen and Satriano \cite{HMPS19}, and partially answering another one posed in \cite{HMPS19}. The first problem is if the problem of deciding whether $\chi_{alg}(G)=4$ is decidable.

\begin{corollary}
It is undecidable to determine whether $\chi_{alg}(G)=4$.
\end{corollary}

\begin{proof}
Since $\chi_{alg}(G) \leq 4$ for all graphs $G$ \cite{HMPS19}, this problem is equivalent to showing that $\chi_{alg}(G)>3$, and this problem being the negation of determining whether $\chi_{alg}(G) \leq 3$, is undecidable by Corollary \ref{corollary: undecidable}.
\end{proof}

The second problem we address from \cite{HMPS19} is whether the \textbf{locally commuting chromatic number} from \cite{HMPS19}, denoted $\chi_{lc}$, is comparable to any of the chromatic numbers $\chi_t$, $t \in \{loc,q,qa,qc,C^*,hered,alg\}$. For a graph $G$ and a number $c \in \bN$, one considers the universal unital $*$-algebra $\cA(\text{Hom}_{lc}(G,K_c))$ generated by self-adjoint idempotents $e_{a,x}$, $x \in V(G)$, $a \in \{1,...,c\}$, satisfying:
\begin{itemize}
\item $\displaystyle \sum_{a=1}^c e_{a,x}=1$;
\item $e_{a,x}e_{b,x}=0$ if $a\neq b$;
\item $e_{a,x}e_{a,y}=0$ if $x \sim y$ in $G$; and
\item $[e_{a,x},e_{b,y}]=0$ if $x \sim y$ in $G$.
\end{itemize}

One then defines $\chi_{lc}(G)$ to be the smallest number $c$ for which $\cA(\text{Hom}_{lc}(G,K_c))$ is non-trivial. For $3$-colorings, projections corresponding to adjacent vertices automatically commute by Proposition \ref{proposition: key three projection proposition, orthogonal version}. Thus, if $t \in \{loc,q,qa,qc,C^*,hered,alg\}$ and $G$ is a graph with $\chi_t(G) \leq 3$, then $\cA(\text{Hom}_{lc}(G,K_3))$ is non-trivial, and hence $\chi_{lc}(G) \leq 3$. This would suggest that, if $\chi_{lc}$ were to compare in general to any of the other quantum chromatic numbers $\chi_t$, then one would have $\chi_{lc} \leq \chi_t$. While we do not resolve this problem here, our work shows that $\chi_{lc}$ is distinct from all of these chromatic numbers.

\begin{corollary}
There exists a graph $G$ with $\chi_{lc}(G)=3<\chi_{hered}(G)$. In particular, $\chi_{lc} \neq \chi_t$ for all $t \in \{loc,q,qa,qc,C^*,hered\}$.
\end{corollary}

\begin{proof}
Since $\chi_{hered} \leq \chi_t$, it suffices to show the first claim. Example \ref{example: 338 vertices} exhibits a graph $G$ with $\chi_{alg}(G)=3$, but $\chi_{hered}(G) \geq 4$. The above discussion shows that $\chi_{lc}(G)=3$ as well, so we are done.
\end{proof}

We suspect, for each $t \in \{loc,q,qa,qc,C^*,hered\}$, there exists a graph $G$ for which $\chi_t(G)<\chi_{lc}(G)$, but we have not pursued this avenue here. The natural context to consider is when $\chi_t(G) \geq 4$, as local commutativity is automatic for $3$-coloring games. We also note that $\chi_{alg}(K_5)=4<\chi_{lc}(K_5)=5$ \cite{HMPS19}, so $\chi_{lc}$ is a genuinely different chromatic number than all of the others considered here.

\section{Independence and Clique Numbers}\label{section: independence and clique}

In this section, we show that synchronous games can also be transformed into games involving the independence number and clique number of graphs. Unlike the equivalences in the previous section, the constructions for this section are more succinct and rely on a construction of Atserias et al \cite{AMRSSV19}.

\begin{definition}
\cite{AMRSSV19} Let $\cG=(I,O,\lambda)$ be a synchronous non-local game. Then the \textbf{graph of the game} $\cG$, denoted $X(\cG)$, has vertex set $V(X(\cG))=O \times I$ and edge set 
\[E(X(\cG))=\{ ((a,x),(b,y)) \in (O \times I)^2:(a,x) \neq (b,y), \text{ and } \lambda(a,b,x,y) \lambda(b,a,y,x)=0\}.\] In other words, $(a,x) \sim (b,y)$ in $X(\cG)$ if and only if $(a,x) \neq (b,y)$ and either $\lambda(a,b,x,y)=0$ or $\lambda(b,a,x,y)=0$.
\end{definition}

\begin{remark}
Depending on the definition of $\cG$, the edge set of $X(\cG)$ may not capture all of the rules of the game. Indeed, the disallowed $4$-tuples that are missed are those of the form $(a,a,x,x)$. However, if $\lambda(a,a,x,x)=0$, then as in Remark \ref{remark: lambda(a,a,x,x)=1}, we can replace this rule with $\lambda(a,b,x,y)=0$ for some $y \in I \setminus \{x\}$ and for all $b \in O$. The resulting game algebra will still be $*$-isomorphic to $\cA(\cG)$. Hence, as long as $|I| \geq 2$, we will assume without loss of generality that $\lambda(a,a,x,x)=1$ for all $a,x$. By Remark \ref{remark: lambda(a,a,x,x)=1}, if $|I|=1$, then we can enlarge the game, while having a $*$-isomorphic game algebra, so that $|I| \geq 2$. Hence, we may assume that $|I|  \geq 2$ and $\lambda(a,a,x,x)=1$ for all $a,x$.
\end{remark}

We recall that, for all the usual models $t \in \{loc,q,qa,qc,C^*,hered,alg\}$, one defines $\alpha_t(G)$ to be the largest integer $m$ for which there exists a winning $t$-strategy for $\text{Hom}(K_m,\overline{G})$. Then there is the chain of inequalities:
\begin{equation}
\alpha(G)=\alpha_{loc}(G) \leq \alpha_q(G) \leq \alpha_{qa}(G) \leq \alpha_{qc}(G) \leq \alpha_{C^*}(G) \leq \alpha_{hered}(G) \leq \alpha_{alg}(G) \label{independence number inequalities}
\end{equation}
The clique number $\omega(G)$ of the graph $G$ is the largest integer $m$ for which there is a homomorphism $K_m \to G$. In the same way, one defines $\omega_t(G)$ for $t \in \{loc,q,qa,qc,C^*,hered,alg\}$. Evidently the clique number satisfies $\omega_t(G)=\alpha_t(\overline{G})$, so the analogue of (\ref{independence number inequalities}) holds for the clique number.

Due to work in synchronous BCS games, we know that the first two inequalities are not equalities in general. The main result of this section is that any non-local game is hereditarily $*$-equivalent to a game involving the independence number. As a result, the first five quantities in (\ref{independence number inequalities}) are all distinct. (The last two are automatically distinct as a result of \cite{HMPS19}; see Remark \ref{remark: alg independence infinity}).

To start, we need a basic observation on hereditary winning strategies. This result is akin to ``projective packings" in \cite{AMRSSV19}.

\begin{proposition}
\label{proposition: projective packing}
Let $\cG=(I,O,\lambda)$ be synchronous, and let $\{ p_{a,x}\}_{(a,x) \in O \times I}$ be a set of projections in a hereditary unital $*$-algebra $\cA$ such that $p_{a,x}p_{b,y}=0$ for all $a,b,x,y$ with $\lambda(a,b,x,y)=0$. Then $\displaystyle \sum_{(a,x) \in O \times I} p_{a,x} \leq |I|$.
\end{proposition}

\begin{proof}
Note that $p_{a,x}p_{b,x}=0$ for $a \neq b$, since $\lambda(a,b,x,x)=0$. Thus, $P_x=\sum_{a \in O} p_{a,x}$ is an orthogonal projection in $\cA$ for each $x \in I$. It follows that $1-P_x$ is an orthogonal projection in $\cA$. Notice that
\[ |I|-\sum_{(a,x) \in O \times I} p_{a,x}=|I|-\sum_{x \in I} P_x=\sum_{x \in I} (1-P_x) \geq 0,\]
so that $\sum_{a,x} p_{a,x} \leq |I|$.
\end{proof}

As a result of Proposition \ref{proposition: projective packing}, we always have the following upper bound on $\alpha_{hered}(X(\cG))$.

\begin{proposition}
\label{proposition: upper bound on hereditary independence number}
Let $\cG=(I,O,\lambda)$ be a synchronous game with $\lambda(a,a,x,x)=1$ for all $a,x$. Then $\alpha_{hered}(X(\cG)) \leq |I|$.
\end{proposition}

\begin{proof}
Assume $m \in \bN$ and $K_m \to_{hered} \overline{X(\cG)}$; that is, assume that $\alpha_{hered}(X(\cG)) \geq m$. Then there are projections $g_{(a,x),z}$ in a unital $*$-algebra $\cB$, for $1 \leq z \leq m$ and $(a,x) \in O \times I$, such that $\sum_{(a,x) \in O \times I} g_{(a,x),z}=1$ for all $z$, $f_{(a,x),z}f_{(b,y),z}=0$ if $(a,x) \neq (b,y)$, and
\begin{equation}
f_{(a,x),z}f_{(b,y),w}=0 \text{ if }  z \neq w \text{ and either } (a,x)=(b,y) \text{ or } (a,x) \sim_{X(\cG)} (b,y). \label{adjacency for independence game}
\end{equation}
By the above, we have $f_{(a,x),z}f_{(a,x),w}=0$ if $z \neq w$, so the element $F_{(a,x)}=\sum_{z=1}^m f_{(a,x),z}$ is an orthogonal projection in $\cB$. Moreover, if $\lambda(a,b,x,y)=0$, then
\[ F_{(a,x)}F_{(b,y)}=\sum_{z,w=1}^m f_{(a,x),z}f_{(b,y),w}=\sum_{z=1}^m f_{(a,x),z}f_{(b,y),z}+\sum_{z \neq w} f_{(a,x),z}f_{(b,y),w}.\]
The first sum is $0$ by synchronicity. Since $(a,x) \neq (b,y)$ and $\lambda(a,b,x,y)=0$, we have $(a,x) \not\sim_{\overline{X(\cG)}} (b,y)$; hence, the second sum is zero. Hence, $F_{(a,x)}F_{(b,y)}=0$ whenever $\lambda(a,b,x,y)=0$. By Proposition \ref{proposition: projective packing}, it follows that $\sum_{(a,x) \in O \times I} F_{(a,x)} \leq |I|$, but we also have that
\[ \sum_{(a,x) \in O \times I} F_{(a,x)}=\sum_{(a,x) \in O \times I} \sum_{z=1}^m f_{(a,x),z}=\sum_{z=1}^m \sum_{(a,x) \in O \times I} f_{(a,x),z}=\sum_{z=1}^m 1=m,\]
so $m \leq |I|$, completing the proof.
\end{proof}

We can now prove the main theorem of this section.

\begin{theorem}
\label{theorem: independence}
Let $\cG=(I,O,\lambda)$ be a synchronous non-local game with $\lambda(a,a,x,x)=1$ for all $(a,x) \in O \times I$. Then $\cG$ is hereditarily $*$-equivalent to the game $\text{Hom}(K_{|I|},\overline{X(\cG)})$. Moreover, if $t \in \{loc,q,qa,qc,C^*,hered\}$, then $\cG$ has a winning $t$-strategy if and only if $\alpha_t(X(\cG))=|I|$.
\end{theorem}

\begin{proof}
We write $e_{a,x}$, $x \in I$, $a \in O$, for the generators of $\cA(\cG)$. These are self-adjoint idempotents satisfying $\sum_{a \in O} e_{a,x}=1$ for all $x \in I$ and $e_{a,x}e_{b,y}=0$ whenever $\lambda(a,b,x,y)=0$. We write $f_{(a,x),z}$, $a \in O$, $x,z \in I$, for the generators of $\cA(\text{Hom}(K_{|I|},\overline{X(\cG)})$. These are self-adjoint idempotents that satisfy $\sum_{(a,x) \in O \times I} f_{(a,x),z}=1$ for all $1 \leq z \leq m$ and
\begin{align}
f_{(a,x),z}f_{(b,y),z}&=0 \text{ if } (a,x) \neq (b,y), \label{independence bisynchronous}\\
f_{(a,x),z}f_{(b,y),w}&=0 \text{ if } z \neq w \text{ and } (a,x) \not\sim (b,y) \text{ in } \overline{X(\cG)}. \label{independence adjacency}
\end{align}

Notice that the elements $g_{(a,x),z}=\delta_{xz} e_{a,x}$ are projections in $\cA(\cG)$ with \[\sum_{(a,x) \in O \times I} g_{(a,x),z}=\sum_{a \in O} e_{a,z}=1.\]
We also have $g_{(a,x),z}g_{(b,y),z}=0$ for $(a,x) \neq (b,y)$, and if $z \neq w$ and $(a,x)=(b,y)$, then either $x \neq z$ or $y \neq w$, giving $g_{(a,x),z}g_{(b,y),w}=0$. If $z \neq w$ and $(a,x) \sim (b,y)$ in $X(\cG)$, then $\lambda(a,b,x,y)=0$, so that
\[ g_{(a,x),z}g_{(b,y),w}=\delta_{zx}\delta_{wy} e_{a,x}e_{b,y}=0.\]
Therefore, the elements $g_{(a,x),z}$ satisfy the relations in (\ref{independence bisynchronous}) and (\ref{independence adjacency}), so there is a unital $*$-homomorphism $\pi:\cA(\text{Hom}(K_{|I|},\overline{X(\cG)})) \to \cA(\cG)$ such that $\pi(f_{(a,x),z})=\delta_{xz}e_{a,x}$.

For the other direction, we consider the hereditary algebra $\cA^h(\text{Hom}(K_{|I|},\overline{X(\cG)}))$, with generators $h_{(a,x),z}$, $(a,x) \in O \times I$, $z \in I$, satisfying the same relations as $f_{(a,x),z}$, except that $\cA^h(\text{Hom}(K_{|I|},\overline{X(\cG)}))$ is hereditary. Now, suppose that $(a,z),(b,w) \in O \times I$ and $(a,z) \neq (b,w)$. We claim that $h_{(a,x),z}h_{(b,x),w}=0$. Indeed, if $z=w$, then since $(a,z) \neq (b,w)$, this means that $a \neq b$; hence, $(a,x) \neq (b,x)$. Then by synchronicity,
\[ h_{(a,x),z}h_{(b,x),w}=h_{(a,x),z}h_{(b,x),z}=0.\]
If $z \neq w$ and $a=b$, then $(a,x)=(b,x)$; in particular, $(a,x) \not\sim (b,x)$ in $\overline{X(\cG)}$. By the last relation in $\cA(\text{Hom}(K_{|I|},\overline{X(\cG)}))$ (passed to the hereditary quotient), we must have $h_{(a,x),z}h_{(b,x),w}=0$. The last case is when $z \neq w$ and $a \neq b$. As $a \neq b$, it follows that $\lambda(a,b,x,x)=0$, since $\cG=(I,O,\lambda)$ is synchronous. Thus, $(a,x) \sim_{X(\cG)} (b,x)$, so $(a,x) \not\sim_{\overline{X(\cG)}} (b,x)$. By the rules of the independence game, $h_{(a,x),z}h_{(b,x),w}=0$. Therefore, for any $(a,z)$, $(b,w)$ with $(a,z) \neq (b,w)$, we have $h_{(a,x),z}h_{(b,x),w}=0$.

It follows that $P_x=\sum_{a \in O, \, z \in I} h_{(a,x),z}$ is an orthogonal projection in $\cA^h(\text{Hom}(K_{|I|},\overline{X(\cG)}))$. Meanwhile,
\[ \sum_{x \in I} P_x=\sum_{z \in I} \sum_{(a,x) \in O \times I} h_{(a,x),z}=\sum_{z \in I} 1=|I|,\]
so that
\[ \sum_{x \in I} (1-P_x)=\sum_{x \in I} 1 - \sum_{x \in I} P_x=|I|-|I|=0.\]
Since $\cA^h(\text{Hom}(K_{|I|},\overline{X(\cG)}))$ is hereditary, it follows that $1-P_x=0$, so that $P_x=1$. Therefore, if we set $p_{a,x}=\sum_{z \in I} h_{(a,x),z}$, then each $p_{a,x}$ is an orthogonal projection (by the orthogonality relations obtained above) and $\sum_{a \in O} p_{a,x}=\sum_{a \in O, \, z \in I} h_{(a,x),z}=P_x=1$. Lastly, if $\lambda(a,b,x,y)=0$, then
\[ p_{a,x}p_{b,y}=\sum_{z,w \in I} h_{(a,x),z}h_{(b,y),w}=\sum_{z \in I} h_{(a,x),z}h_{(b,y),z}+\sum_{z \neq w} h_{(a,x),z}h_{(b,y),w}=0,\]
where the first sum is $0$ by the synchronous rule in the homomorphism game, and the second sum is $0$ since $z \neq w$, but $(a,x) \not\sim (b,y)$ in $\overline{X(\cG)}$.

Therefore, there is a unital $*$-homomorphism $\rho:\cA(\cG) \to \cA^h(\text{Hom}(K_{|I|},\overline{X(\cG)}))$ such that $\rho(e_{a,x})=h_{a,x}$. Thus, the games $\cG$ and $\text{Hom}(K_{|I|},\overline{X(\cG)})$ are hereditarily $*$-equivalent. The final claim of the theorem follows since we already have $\alpha_{hered}(X(\cG)) \leq |I|$, by Proposition \ref{proposition: upper bound on hereditary independence number}.
\end{proof}

\begin{corollary}
\label{corollary: independence gap}
The quantities $\alpha_{loc}$, $\alpha_q$, $\alpha_{qa}$, $\alpha_{qc}$ and $\alpha_{C^*}$ are all distinct.
\end{corollary}

\begin{proof}
It is already known that $\alpha_{loc}$ and $\alpha_q$ are distinct, using the graph of the magic square game that originated in \cite{Mer90} (see \cite{AMRSSV19} for the graph). As there is a syncBCS game with winning $qa$ strategy but no winning $q$ strategy, the work of \cite{KPS18} obtains a graph that separates $\alpha_q$ and $\alpha_{qa}$. The fact that $\alpha_{qa}$ and $\alpha_{qc}$ are distinct follows from Theorem \ref{theorem: independence} and the existence of a synchronous game with winning $qc$ strategy, but no winning $qa$ strategy \cite{JNVWY20}. A synchronous game with non-zero $C^*$-algebra, but no tracial state, has been exhibited by C. Paddock and W. Slofstra \cite{PS21}. From this fact, one obtains a graph $G$ with $\alpha_{qc}(G)<\alpha_{C^*}(G)$.
\end{proof}

\begin{corollary}
The quantities $\omega_{loc}$, $\omega_q$, $\omega_{qa}$, $\omega_{qc}$ and $\omega_{C^*}$ are all distinct.
\end{corollary}

\begin{proof}
The $t$-clique number of $G$ is the maximal integer $m$ such that $K_m \to_t G$, so $\omega_t(G)=\alpha_t(\overline{G})$. Taking $G=\overline{X(\cG)}$ and applying Corollary \ref{corollary: independence gap} yields the result.
\end{proof}

\begin{remark}\label{remark: t=alg}
Theorem \ref{theorem: independence} is false for $t=alg$ if $|I| \geq 4$. The unital $*$-homomorphism $\cA(\text{Hom}(K_{|I|},\overline{X(\cG)})) \to \cA(\cG)$ from the proof of Theorem \ref{theorem: independence} still exists, but the reverse direction requires the hereditary algebra $\cA^h(\text{Hom}(K_{|I|},\overline{X(\cG)}))$ to be non-zero. The key step where the algebra is required to be hereditary is in the final step, where we construct $|I|$ self-adjoint idempotents that sum to zero. In a hereditary $*$-algebra, this forces each idempotent to be zero, but this is not true in a unital $*$-algebra, as the universal unital $*$-algebra generated by four self-adjoint idempotents $p_1,p_2,p_3,p_4$ with the relation $p_1+p_2+p_3+p_4=0$ is non-trivial \cite{BES94,ST02}. If $|I|=3$, then the above result extends to the algebraic model since the only way to write $0$ as a sum of three idempotents is with each idempotent equal to zero (see Proposition \ref{proposition: three projections summing to zero}).
\end{remark}

\begin{remark}
\label{remark: alg independence infinity}
If $|I| \geq 4$, and if $\alpha_t(X(\cG)) \geq 4$ for some $t \in \{loc,q,qa,qc,C^*,hered,alg\}$, then $K_4 \to_t \overline{X(\cG)}$, so that $K_4 \to_{alg} \overline{X(\cG)}$. But by \cite{HMPS19}, $K_m \to_{alg} K_4$ for any $m \geq 4$. We claim that $K_m \to_{alg} \overline{X(\cG)}$. We write the generators of $\cA(\text{Hom}(K_4,\overline{X(\cG)}))$ as $e_{(a,x),v}$ for $(a,x) \in O \times I$ and $1 \leq v \leq 4$, and the generators of $\cA(\text{Hom}(K_m,K_4))$ by $f_{z,v}$ for $1 \leq z \leq m$ and $1 \leq v \leq 4$. Define
\[ g_{(a,x),z}=\sum_{v=1}^4 f_{z,v} \otimes e_{(a,x),v}.\]
Then $g_{(a,x),z}$ is self-adjoint, and
\begin{align*}
g_{(a,x),z}^2&=\sum_{v,w=1}^4 f_{z,v}f_{z,w} \otimes e_{(a,x),v}e_{(a,x),w} \\
&=\sum_{v=1}^4 f_{z,v}\otimes e_{(a,x),v}=g_{(a,x),z},
\end{align*}
since $f_{z,v}f_{z,w}=0$ if $v \neq w$. We also have
\[ \sum_{(a,x) \in O \times I} g_{(a,x)}=\sum_{v=1}^4\sum_{a,x \in O \times I} f_{z,v} \otimes e_{(a,x),v}=\sum_{v=1}^4 f_{z,v} \otimes 1=1 \otimes 1,\]
and, if $1 \leq z,w \leq m$ with $z \neq w$ (so that $z \sim w$ in $K_m$), and if $(a,x) \not\sim (b,y)$ in $\overline{X(\cG)}$ (so that either $(a,x)=(b,y)$ or $\lambda(a,b,x,y)=0$), then since $e_{(a,x),u}e_{(b,y),v}=0$ for $u \neq v$, we have
\begin{align*}
g_{(a,x),z}g_{(b,y),w}&=\sum_{u,v=1}^4 f_{z,u}f_{w,v} \otimes e_{(a,x),u}e_{(b,y),v} \\
&=\sum_{v=1}^4 f_{z,v}f_{w,v} \otimes e_{(a,x),v}e_{(b,y),v} \\
&=0,
\end{align*}
since $z \neq w$. Therefore, by the universal property of $\cA(\text{Hom}(K_m,\overline{X(\cG)}))$, we have \[\cA(\text{Hom}(K_m,\overline{X(\cG)})) \to \cA(\text{Hom}(K_m,K_4)) \otimes \cA(\text{Hom}(K_4,\overline{X(\cG)})).\] 
As both algebras in the tensor product are non-trivial, we have $K_m \to_{alg} \overline{X(\cG)}$, so that $\alpha_{alg}(X(\cG)) \geq m$. This holds for any $m \geq 4$, so $\alpha_{alg}(X(\cG))=\infty$. Note that this can happen whether $\cG$ has a winning hereditary strategy or not. For example, the trivial synchronous game with $|I|=4$ and $|O|=2$, with rule function $\lambda(a,b,x,x)=\delta_{ab}$ and $\lambda(a,b,x,y)=1$ for $x \neq y$, has a winning loc strategy. Thus, $\alpha_{loc}(X(\cG))=4$, so the above argument shows that $\alpha_{alg}(X(\cG))=\infty$.

On the other hand, the game $\cG=\text{Hom}(K_5,K_4)$ has a winning algebraic strategy \cite{HMPS19}, but no winning hereditary strategy. By Theorem \ref{theorem: independence}, $\alpha_{hered}(X(\cG)) \leq 4$, since $|I|=5$. The proof of Theorem \ref{theorem: independence} shows that $\cA(\text{Hom}(K_{|I|},\overline{X(\cG)})) \to \cA(\cG)$, so $K_5 \to_{alg} \overline{X(\cG)}$. Thus, $\alpha_{alg}(X(\text{Hom}(K_5,K_4))) \geq 4$, forcing $\alpha_{alg}(X(\text{Hom}(K_5,K_4)))=\infty$.
\end{remark}

The situation is different when $|I| \leq 3$. In that case, Theorem \ref{theorem: independence} can be extended to the algebraic setting, in that $\cG$ will have a non-zero game algebra if and only if $\alpha_{alg}(X(\cG))=3$, as in Remark \ref{remark: t=alg}.

\section*{Acknowledgements}

The author would like to thank Michael Brannan, Laura Man\v{c}inska, Vern Paulsen, David Roberson, Ivan Todorov and Michele Torielli for valuable discussions. We would also like to thank Arthur Mehta and William Slofstra for pointing us to the application of Corollary \ref{corollary: coloring equivalence} to perfect zero knowledge.

\end{document}